\documentclass[12pt]{iopart} 

\expandafter\let\csname equation*\endcsname\relax
\expandafter\let\csname endequation*\endcsname\relax
\usepackage{amsmath}
\usepackage{amsthm}
\usepackage{amssymb}
\usepackage{amsfonts}
\usepackage{amstext}
\usepackage{amsopn}
\usepackage{thm-restate}
\usepackage{graphicx}
\usepackage[hidelinks]{hyperref}
\usepackage{bbm}
\usepackage{cleveref}
\crefname{equation}{}{}
\crefname{figure}{figure}{figures}
\usepackage{iopams}

\newcommand{\bra}[1]{\left\langle #1 \right\rvert}
\newcommand{\ket}[1]{ \left\lvert #1\right\rangle}

\newcommand{\ketbra}[2]{\left\lvert #1 \rangle \! \langle #2 \right\rvert}

\newcommand{\tV}{\vert\kern-0.25ex\vert\kern-0.25ex\vert}

\newcommand{\set}[1]{\{#1\}}

\renewcommand{\tr}[1]{\operatorname{Tr}\left(#1\right)}

\DeclareMathOperator{\id}{id}

\DeclareMathOperator{\End}{End}

\newcommand{\cA}{\ensuremath{\mathcal{A}}}

\newcommand{\cD}{\ensuremath{\mathcal{D}}}

\newcommand{\cH}{\ensuremath{\mathcal{H}}}

\newcommand{\cM}{\ensuremath{\mathcal{M}}}

\newcommand{\cS}{\ensuremath{\mathcal{S}}}
\newcommand{\cT}{\ensuremath{\mathcal{T}}}
\newcommand{\cU}{\ensuremath{\mathcal{U}}}

\newcommand{\bbC}{\ensuremath{\mathbb{C}}}

\newcommand{\bbI}{\ensuremath{\mathbb{I}}}

\newcommand{\bbN}{\ensuremath{\mathbb{N}}}

\newcommand{\bbP}{\ensuremath{\mathbb{P}}}

\newcommand{\bbR}{\ensuremath{\mathbb{R}}}

\newcommand{\calU}{\mathcal{U}}

\newcommand{\calD}{\mathcal{D}}
\newcommand{\lb}{\left(}
\newcommand{\rb}{\right)}
\newcommand{\Dkl}[2]{D\lb#1\lvert\rvert#2\rb}

\theoremstyle{plain}% default
\newtheorem{thm}{Theorem}
\newtheorem{lem}[thm]{Lemma}

\newtheorem{cor}[thm]{Corollary}
\theoremstyle{definition}
\newtheorem{defn}[thm]{Definition}

\theoremstyle{remark}

\begin{document}

\title[Approximate randomized benchmarking for finite groups]{Approximate randomized benchmarking \\ for finite groups}

\author{D~S Fran\c{c}a and A~K Hashagen}
\eads{\mailto{dsfranca@mytum.de}, \mailto{hashagen@ma.tum.de}}
\address{Department of Mathematics, Technical University of Munich, Germany}
\vspace{10pt}
\begin{indented}
\item[]March 2018
\end{indented}

\begin{abstract}
We investigate randomized benchmarking in a general setting with quantum gates that form a representation, not necessarily an irreducible one, of a finite group. We derive an estimate for the average fidelity, to which experimental data may then be calibrated. Furthermore, we establish that randomized benchmarking can be achieved by the sole implementation of quantum gates that generate the group as well as one additional arbitrary group element.
In this case, we need to assume that the noise is close to being covariant. This yields a more practical approach
to randomized benchmarking. Moreover, we show that randomized benchmarking is stable
with respect to approximate Haar sampling for the sequences of gates. This opens up the
possibility of using Markov chain Monte Carlo methods to obtain the random sequences
of gates more efficiently. We demonstrate these results numerically using the well-studied
example of the Clifford group as well as the group of monomial unitary matrices. For the
latter, we focus on the subgroup with nonzero entries consisting of n-th roots of unity, which
contains T gates.
\end{abstract}

\vspace{2pc}

\submitto{\jpa}
\maketitle

%%%%%%%%%%%%%%%%%%%%%%%%%%%%%%%%%%%%%%%%%%%%%%%%%%%%%%
\section{Introduction}
\label{sec:Introduction}
%%%%%%%%%%%%%%%%%%%%%%%%%%%%%%%%%%%%%%%%%%%%%%%%%%%%%%
One of the main obstacles to build reliable quantum computers is the need to implement
quantum gates with high fidelity.
Therefore, it is key to develop techniques to estimate the quality of quantum gates and thus
certify the quality of a quantum computer.
To this end, one could perform tomography for the underlying noise in the implementation
and in principle obtain a complete description of it~\cite{Poyatos_1997,Chuang_1997}. 
However, in general, the number of measurements necessary to estimate for a complete tomography of the noise scales exponentially
with the system size and is not a practical solution to the problem.
Thus, it is vital to develop techniques to estimate the level of noise in systems more efficiently, even if we only
obtain partial information.

Randomized benchmarking (RB) is a protocol to estimate the average fidelity of a set of quantum gates forming a representation of a group~\cite{Knill_2008, Emerson_2007, Levi_2007, Emerson_2005}.
The very important case of Clifford gates has already been widely studied and some rigorous results that show its efficiency under
some noise scenarios are available~\cite{Wallman_Flammia_2014,Helsen_2017}, such as when the noise is independent of the gate and time.
Besides its efficiency, another highlight of the protocol is that it is robust against state preparation and measurement errors.
This makes it very attractive from an experimental point of view
and its applicability was demonstrated successfully~\cite{Chow_2009, Ryan_2009, Olmschenk_2010, Brown_2011, Gaebler_2012, Barends_2014, Xia_2015, Muhonen_2015, Asaad_2016}. 

In this work, we show how to extend these protocols to gates that are representations of a finite group\footnote{Most of the results in this work can easily be extended to compact groups. However, as it is not clear that implementing the RB protocol for compact groups is relevant for applications and given that this would make some proofs less accessible, we restrict to finite groups here.}; these must not necessarily be irreducible or form a $2$-design. 
Although other works, such as~\cite{Hashagen_2018, Brown_Eastin_2018, Cross_2016, CarignanDugas_2015}, already 
extended the protocol to other specific groups of interest, we focus on showing how to estimate the average fidelity based on
properties of the particular representation at hand for arbitrary finite groups.
To this end, we investigate the structure of quantum channels that are covariant under a unitary 
representation of a group and derive formulas
for their average fidelity in terms of their spectra. 
We then show that one can use RB to estimate the average fidelity of these gates
under the assumption that they are subject to time and gate independent noise.

In order for this procedure to be efficient, it is necessary that we may multiply, invert and sample uniformly distributed 
elements of the group efficiently and
that the given representation does not decompose into too many irreducible unitary representations, as we will discuss in more detail later.
This is the case for the well-studied case of Cliffords.

The usual RB protocol assumes that we can implement sequences of gates that are sampled from the Haar distribution
of the group~\cite{Knill_2008, Emerson_2007, Levi_2007, Emerson_2005}.
We further generalize the RB protocol by showing that it is possible to implement sequence gates that are
approximately Haar distributed instead.
Therefore, it is possible to use Markov chain Monte Carlo methods to obtain the samples, potentially more efficiently. 
This result is of independent interest to the RB literature, as it shows that the protocol is stable against small errors in the sampling.

Moreover, we show how one can perform RB by just implementing gates that generate the group and 
one additional random element from the group at each round of the protocol. Thus, this last gate will generally not be an element of the generators.
Mostly considering generators provides a more natural framework to the protocol, 
as often one is only able to implement a certain number of gates that generate the group and must,
 therefore, decompose the gates into generators. However, this protocol works under the assumption that
the noise affecting the gates is already close to being covariant with respect to (w.r.t.) the group and not for arbitrary quantum channels, as
in the usual setting. Moreover, we still need the ability to implement one gate which might not be contained in the set of generators and still assume
that the same quantum channel that describes the noise on the generators also describes the noise on this gate. 
To illustrate our techniques, we apply them to subgroups of the monomial unitary matrices, i.e. products of $d-$dimensional permutation and diagonal unitary matrices. These can
be seen as a generalization of stabilizer groups~\cite{VandenNest_2011}. 
We focus on the subgroup of monomial unitary matrices whose nonzero entries are roots of unity. We show that 
we only need to estimate two parameters and multiplying and inverting elements of it 
can be done in time
$\Or(d)$. %This makes these subgroups a natural candidate to apply our techniques.
Moreover, they include the $T$-gate, which is known
to form a universal set for quantum computation together with the Clifford gates~\cite{Nielsen_2009}. 
Therefore, one can use the protocol described here to estimate the noise from 
$T$-gates more efficiently.
We make numerical simulations for our protocol and these subgroups and show that it is able to reliably estimate 
the average gate fidelity.
Moreover, we numerically compare our techniques based on approximate Haar samples and implementation of generators to the usual
protocol for Cliffords and show that the three yield indistinguishable results in the high fidelity regime.

This paper is structured as follows:
we start by fixing our notation and reviewing basic results on Markov chains and covariant quantum channels; needed in \cref{sec:Prelim}. 
In \cref{sec:Fidelities} we derive the average fidelity of quantum channels in terms of their spectra
and we give basic results on the decay of the probability of measurement outcomes under covariant quantum channels. These form the basis
for the RB protocol for general groups, which we discuss and analyze in~\cref{sec:RBprotocol}.
In~\cref{sec:approxtwirls} we prove that it is also possible to implement the protocol using approximate samples.
We then discuss the generalized RB protocol based on implementing random sequences of gates that generate the group
in~\cref{sec:Generators}. In this section, we also discuss the conditions under which this protocol applies.
Finally, in~\cref{sec:Numerics}, we apply our techniques to the subgroup of monomial unitary matrices and perform numerical experiments 
for it. In the same section, we also compare numerically the RB protocols developed here with the usual one in the case
of the Clifford group.

%%%%%%%%%%%%%%%%%%%%%%%%%%%%%%%%%%%%%%%%%%%%%%%%%%%%%%
\section{Notation and Preliminaries}
\label{sec:Prelim}
%%%%%%%%%%%%%%%%%%%%%%%%%%%%%%%%%%%%%%%%%%%%%%%%%%%%%%

We will be interested in finite dimensional quantum systems.
Denote by $\cM_d$  the space of $d\times d$ complex matrices.
We will denote by $\cD_d$ the set of $d$-dimensional quantum states, i.e., positive semi-definite matrices
$\rho\in\cM_d$ with trace $1$. 
We will call a linear map $T:\cM_d\to\cM_{d'}$ a quantum channel if it is trace preserving and completely positive. We will denote the adjoint of a quantum channel $T$ with
respect to the Hilbert-Schmidt scalar product by $T^*$.
We will call a collection of positive semidefinite matrices $\{E_i\}_{i=1}^{l}$ a positive operator valued measure~(POVM) 
if the POVM elements $E_i$, called effect operators, sum up to the identity.
Throughout this paper, we will use the channel-state duality that 
provides a one-to-one correspondence between a quantum channel $T: \cM_{d} \to \cM_{d}$ and 
its Choi-Jamiolkowski state $\tau_T \in \cM_{d^2}$ obtained by letting $T$ act on half of a maximally entangled state, i.e.,
\begin{equation}
\tau_T:= \left( T \otimes \id_d \right) \left( \ketbra{\Omega}{\Omega} \right),
\label{eq:ChoiJami}
\end{equation}
where $\ketbra{\Omega}{\Omega} \in \cM_{d^2}$  is a maximally entangled state, that is,
\begin{equation}
\ketbra{\Omega}{\Omega}= \frac{1}{d}\sum_{i,j=1}^d \ketbra{ii}{jj},
\label{eq:Omega}
\end{equation}
where $\{\ket{i}\}_{i=1}^d$ is an orthonormal basis in $\bbC^d$.
Please refer to~\cite{Heinosaari_Ziman_2012} for more on these concepts.  
To measure the distance between two states we will use the Schatten $1-$norm for $A\in\cM_d$, denoted by $\|\cdot\|_1$
and given by
\begin{align}
\|A\|_1:=\tr{(A^\dagger A)^{\frac{1}{2}}}, 
\end{align}
where $\dagger$ denotes the adjoint.
Then, given two states $\rho,\sigma\in\cD_d$, their trace distance is given by
$
\|\rho-\sigma\|_1 / 2 
$.
This norm on $\cM_d$ induces a norm on linear operators $\Phi:\cM_d\to\cM_d$ through
\begin{align}
\|\Phi\|_{1\to1}:=\sup\limits_{X\in\cM_d, X\not=0}\frac{\|\Phi(X)\|_1}{\|X\|_1}.
\end{align}
Given a random quantum channel $T:\cM_d\to\cM_d$, we will denote its expectation value by $\mathbb{E}(T)$.

We will also need some basic facts from the representation theory of finite groups. We refer to e.g.~\cite{Simon_1996} for more on this and the proofs
of the statements we use here. We will be particularly interested in the commutant of the algebra generated by the group. To this end we introduce:
\begin{defn}[Commutant]
Let $\cA$ be an algebra of operators on a Hilbert space $\cH$. Then the commutant $\cA'$ of $\cA$ is defined by
\begin{equation}
\cA' := \left\{ B \middle\vert BA =AB \text{ for all } A \in \cA \right\}.
\label{eq:commutant}
\end{equation}
\end{defn}
Recall that a function $U:G\to\cM_d$ is called a unitary representation of a finite group $G$ on a finite-dimensional Hilbert space $\cH \simeq \bbC^d$ if we have for all $g_1,g_2\in G$ that
$U_{g_1}U_{g_2}=U_{g_1g_2}$. 
We will denote the unitary corresponding to $g$ by $U_g$.
From basic results of representation theory, we know that there exists distinct $\alpha_1, \ldots, \alpha_k \in \hat{G}$, where $\hat{G}$ 
denotes the set of equivalence classes of irreducible unitary representations~(irreps), 
such that the unitary representation can be written as a direct sum of irreps, 
i.e. $U \cong \oplus U^{\alpha_i}\otimes \bbI_{m_\alpha}$ with $m_\alpha > 0$ denoting the degeneracy of the $\alpha_i$-th irrep.
The structure of the commutant 
is then described in the following theorem.
\begin{thm}[{\cite[Theorem IX.11.2]{Simon_1996}}]\label{thm:decompalgebra}
Let $U$ be a unitary representation of a finite group $G$ on $\cH$. Write 
$\cH = \oplus_{\alpha\in\hat{G}} \left( \bbC^{d_{\alpha}} \otimes \bbC^{m_\alpha}\right)$ so that 
$U_g = \oplus^k_{i=1} U_g^{\alpha_i} \otimes \bbI_{m_\alpha}$ with $\left\{\alpha_i \right\}^k_{i=1}$ distinct elements in $\hat{G}$. 
Let $\cA(U)$ be the algebra of operators generated by the $\left\{ U_g \right\}_{g \in G}$, and $\cA(U)'$ its commutant. Then
\begin{subequations}
\begin{align}
\cA(U) = \left\{ \oplus^k_{i=1} A_i \otimes \bbI_{m_\alpha} \middle\vert A_i \in \cM_{d_{\alpha_i}} \right\}, \\
\cA(U)' = \left\{ \oplus^k_{i=1}  \bbI_{d_{\alpha_i}} \otimes B_i \middle\vert B_i \in \cM_{m_\alpha} \right\}.
\end{align}
\end{subequations}
\end{thm}

Given a finite group $G$, we will call the uniform probability distribution on it its Haar measure. For a proof of its existence and basic properties, 
we refer to~\cite[Section VII.3]{Simon_1996}. 
Given some unitary representation $U:G\to\cM_d$, we call the function $\chi:G\to\mathbb{C}$ given by $g\mapsto\tr{U_{g}}$ the character of the representation.
We will denote the character of an irreducible representation $\alpha\in\hat{G}$ by $\chi^\alpha$ and remark that one can find the decomposition
in \cref{thm:decompalgebra} through characters~\cite[Section III.2]{Simon_1996}.

%%%%%%%%%%%%%%%%%%%%%%%%%%%%%%%%%%%%%
\subsection{Covariant Quantum Channels and Twirls}
\label{sec:twirls}
%%%%%%%%%%%%%%%%%%%%%%%%%%%%%%%%%%%%%
The definition of covariance of quantum channels is central to the study of their symmetries and will be one of the building
blocks of the generalized RB protocol:
\begin{defn}[Covariant quantum channel {\cite{Mendl_Wolf_2009}}]
A quantum channel $T: \cM_d \to \cM_d$ is covariant w.r.t. a 
unitary representation $U: G\to\cM_d$ of a finite group $G$, if for all $g \in G$
\begin{equation}
T\left( U_g \cdot U_g^\dagger \right) = U_g T\left( \cdot \right) U_g^\dagger. 
\label{eq:CovariantQC}
\end{equation}
\end{defn}

In general, one allows different unitary representations of the group in the input and output of the channel in the definition of covariance,
but here we will restrict to the case when we have the same unitary representation.
There are many different and equivalent characterizations of covariance. Here we mention that covariance is
equivalent to the Choi-Jamiolkowski state $\tau_T$ commuting with $U_g\otimes  \bar{U}_g$ for all $g\in G$.
To see this, note that given a unitary representation $U$ of $G$ we may define its adjoint representation 
$\cU:G \to \End (\cM_d)$ through its action on any $X\in \cM_d$ by conjugation,
\begin{equation}
\cU_g(X) = U_g X U_g^\dagger.
\end{equation}
Through the Choi-Jamiolkowski isomorphism, it is easy to see that the adjoint representation 
is equivalent to  the unitary representation $U_g \otimes \bar{U}_g \in \cM_{d^2}$. 
As we can rephrase \cref{eq:CovariantQC} as $T$ commuting with the adjoint representation, this
translates to the Choi-Jamiolkowski state commuting with $U_g \otimes \bar{U}_g$.
This means in particular that we may use structural theorems, like \cref{thm:decompalgebra},
to investigate covariant channels, as covariance implies that the channel is  in the commutant of the adjoint representation. 
\begin{thm}\label{thm:structurecov}
Let $T: \cM_d \to \cM_d$ be a quantum channel that is covariant w.r.t. a unitary representation
$U$ of a finite group $G$
and let $\oplus_{\alpha\in\hat{G}} \left( \bbC^{d_{\alpha}} \otimes \bbC^{m_\alpha}\right)$
be the decomposition of the underlying Hilbert space into irreps $\alpha$ of $G$ with multiplicity $m_\alpha$ for the unitary representation $U\otimes\bar{U}$.
Then:
\begin{equation}
 T=\oplus_{\alpha\in\hat{G}}  \bbI_{d_{\alpha}} \otimes B_{\alpha} 
\end{equation}
with $B_\alpha\in\cM_{m_\alpha}$. 
\end{thm}

\begin{proof} 
As $T$ is covariant, it must  be an element of the commutant of the adjoint representation, i.e. $T\in\cA(\cU)'$. 
The decomposition then follows from \cref{thm:decompalgebra}.
\end{proof}

This decomposition further simplifies
when no multiplicities in the decomposition of the unitary representations into its irreducible components are present.
We call such channels irreducibly covariant.
Here we briefly mention some of the results of~\cite{Mozrzymas_2017}, where the structure of such channels is investigated. 

\begin{thm}[{\cite[Theorem 40]{Mozrzymas_2017}}]\label{thm:structirredcov}
A quantum channel $T: \cM_d \to \cM_d$ is irreducibly covariant w.r.t. an irrep $U:G\to\cM_d$
of a finite group $G$ if and only if it has a decomposition of the following form:
\begin{equation}
T = l_{\id} P^{\id} + \sum_{\alpha \in \hat{G}, \alpha \neq \id} l_\alpha P^\alpha,
\label{eq:DecompositionQC}
\end{equation}
with $l_{\id}=1$, $l_\alpha \in \bbC$ and where $P^{\id}, P^{\alpha}: \cM_d \to \cM_d$ are projectors defined as
\begin{equation}\label{def:projec}
P^\alpha (\cdot) = \frac{\chi^\alpha(e)}{|G|} \sum_{g \in G} \chi^\alpha \left( g^{-1}\right) U_g \cdot U_g^\dagger,
\end{equation}
with $\alpha \in \hat{G}$ and $e\in G$ the identity of the group. They have the following properties:
\begin{equation}
P^\alpha P^\beta = \delta_{\alpha \beta}P^\alpha, \qquad
(P^\alpha)^\ast = P^\alpha \ \ \text{ and} \qquad
\sum_{\alpha \in \hat{G}} P^\alpha = \id_d,  
\end{equation}
where $\id_d:\cM_d \to \cM_d$ is the identity map and
the coefficients $l_\alpha$ are the eigenvalues of the quantum channel $T$.
\end{thm}
That is, in the case of an irreducibly covariant channel we can also write down the projections onto different eigenspaces and
diagonalize the channel.

One of the most important concepts in this paper is that of the twirl of a channel.
\begin{defn}[Twirl]
Let $T:\cM_d\to\cM_d$ be a quantum channel, $G$ a finite group with Haar measure $\mu$
and $U:U\to\cM_d$ a unitary representation of $G$. We define the twirl of $T$ w.r.t. $G$, denoted by $\mathcal{T}(T):\cM_d\to\cM_d$, as
\begin{equation}
\mathcal{T}(T)(\cdot)=\int_G \cU_g^*\circ T\circ \cU_g(\cdot)\rmd \mu. 
\end{equation}
\end{defn}
Strictly speaking the twirled channel, of course, depends on the particular group and unitary representation at hand. However, 
we will
omit this in the notation, as the group in question should always be clear from context.
It is then easy to show that $\mathcal{T}(T)$ is a quantum channel that is covariant w.r.t. this representation.

%%%%%%%%%%%%%%%%%%%%%%%%%%%%%%%%%%%%%%%%%%%%%%%%%%%%%%
\subsection{Random Walks on Groups}
\label{sec:markov}
%%%%%%%%%%%%%%%%%%%%%%%%%%%%%%%%%%%%%%%%%%%%%%%%%%%%%%
We will need some basic tools from the field of random walks on groups to motivate and explain our protocol to perform
RB with generators or with approximate samples. Therefore, we review these basic concepts here and 
refer to e.g.~\cite[Chapter 2.6]{markovmixing} for more details and proofs.
Given a finite group $G$, we denote the set of probability measures on $G$
by $\mathcal{P}(G)$. If $X,Y$ are two independent random variables on $G$ with distributions $\mu,\nu\in\mathcal{P}(G)$,
respectively, we denote their joint distribution on $G\times G$ by $\mu\otimes\nu$. Analogously, we will denote 
the joint distribution of $Y_1,\ldots,Y_n$ i.i.d. variables with distribution $\nu$ by $\nu^{\otimes n}$ and the $m$-fold
Cartesian product of $G$ with itself by $G^m$.
The random
walk on G with increment distribution $\nu$ is defined as follows: 
it is a Markov chain with
state space $G$. Given that the current state $X_n$ of the chain is $g$, the next state $X_{n+1}$
is given by multiplying the current state on the left by a
random element of $G$ selected according to $\nu$.
That is, we have
\begin{equation}
P(X_{n+ 1}=g_2|X_{n}=g_1)=\nu\lb g_2g_1^{-1}\rb.
\end{equation}
Another way of tracking the transition probabilities for these chains is through the transition matrix of the chain, $\pi$.
For $g_1,g_2\in G$, this matrix is defined as
\begin{equation}
 \pi(g_1,g_2)=\nu\lb g_2g_1^{-1}\rb.
\end{equation}
If $X_0$ is distributed according to $\mu\in\mathcal{P}(G)$, we have that the distribution of $X_n$ is given by
$\pi^n\mu$, where we just expressed $\mu$ as a vector in $\bbR^{|G|}$.
We recall the following fundamental result about random walks on groups:
\begin{thm}
Let $G$ be a finite group and $A$ be a set of generators of $G$ that is closed under inversion.
Moreover, let $\nu$ be the uniform distribution on $A$ and 
$X_1,X_2,\ldots$ be a random walk with increment distribution $\nu$.
Then the distribution of $X_n$ converges to the Haar distribution on $G$ as $n\to\infty$.
\end{thm}
\begin{proof}
We refer to e.g.~\cite[Section 2.6.1]{markovmixing} for a proof and more details on this.
\end{proof}
Given a generating subset $A$  of $G$ that is closed under inverses and $\nu$ the uniform distribution on $A$, 
we will refer to the random walk with increment $\nu$ as the random walk generated by $A$.
This result provides us with an easy way of obtaining samples which are approximately Haar distributed if we have a set of generators
by simulating this random walk for long enough. 
The speed of this convergence is usually quantified in the total variation distance. Given two probability measure $\mu,\nu$ on
$G$, we define their total variation distance to be given by:
\begin{align}
\label{eqn:TVdistance}
\|\mu-\nu\|_1:=\frac{1}{2}\sum\limits_{g\in G}|\mu(g)-\nu(g)|. 
\end{align}
We then define the mixing time of the random walk as follows:
\begin{defn}[Mixing Time of Random Walk]
Let $G$ be a finite group and $A$ a set of generators closed under inverses and $\mu$ be the Haar measure 
on the group. For $\epsilon>0$, the mixing time of the chain generated by $A$, $t_1(\epsilon)$, is defined as 
\begin{equation}
 t_1(\epsilon):=\inf\{n\in\bbN|\forall \nu\in\mathcal{P}(G):\|\pi^n\nu-\mu\|_1\leq\epsilon\}.
\end{equation}
\end{defn}
We set $t_{\text{mix}}$ to be given by $t_1(4^{-1})$, as this is a standard choice in literature~\cite[Section 4.5]{markovmixing}. One can then show that $t_1(\epsilon)\leq\lceil\log_2\lb\epsilon ^{-1}\rb \rceil t_{\text{mix}}$
(see~\cite[Section 4.5]{markovmixing} for a proof).
There is a huge literature devoted to determining the mixing time of random walks on groups and we refer to~\cite{saloff2004random}
and references therein for more details.
For our purposes it will be enough to note that in most cases we have that $t_1(\epsilon)$ scales logarithmically
with $\epsilon^{-1}$ and $|G|$. 
Another distance measure which is quite useful in the study of convergence of random variables is the relative entropy $D$.
For two probabilities measures $\mu,\nu$ on $\{1,\ldots,d\}$ we define their relative entropy to be
\begin{align}
D(\mu||\nu):=
\begin{cases}
\sum\limits_{i=1}^d\mu(i)\log\lb\frac{\mu(i)}{\nu(i)}\rb,\quad\textrm{if }\mu(i)=0\textrm{ for all }i\textrm{ s.t. }\nu(i)=0,\\
+\infty,\quad\textrm{else.}
\end{cases}
\end{align}
One of its main properties is that for $\mu,\nu\in\mathcal{P}(G)$ we have~\cite{Wehrl_1978}
\begin{align}
D(\mu^{\otimes n}||\nu^{\otimes n})=nD(\mu||\nu). 
\end{align}

%%%%%%%%%%%%%%%%%%%%%%%%%%%%%%%%%
\section{Fidelities}
\label{sec:Fidelities}
%%%%%%%%%%%%%%%%%%%%%%%%%%%%%%%%%
Given a quantum channel $T:\cM_d\to\cM_d$ and a unitary channel $\mathcal{U}:\cM_d\to\cM_d$,
 the average fidelity between them is defined as
\begin{align}
F(T,\mathcal{U})=\int\tr{T(\ketbra{\psi}{\psi})\mathcal{U}(\ketbra{\psi}{\psi})}\rmd \psi, 
\end{align}
where we are integrating over the Haar measure on quantum states.
In case $\mathcal{U}$ is just the identity, we refer to this quantity as being the average fidelity of the channel and denote
it by $F(T)$.
As shown in~\cite{Nielsen_2002}, the average fidelity of a channel is a simple function of its entanglement fidelity, given by
\begin{align}
F_e(T)=\tr{T\otimes\id\lb\ketbra{\Omega}{\Omega}\rb\ketbra{\Omega}{\Omega}},
\end{align}
with $\ketbra{\Omega}{\Omega}$ the maximally entangled state.
One can then show that 
\begin{align}
F(T)=\frac{dF_e(T)+1}{d+1}. 
\end{align}
Thus, we focus on estimating the entanglement fidelity instead of estimating the average fidelity. This can be seen
to be just a function of the trace of the channel and the dimension, as we now show.
\begin{lem}\label{lem:entangfidprojec}
Let $T:\cM_d\to\cM_d$ be a quantum channel.
Then $F_e(T)=d^{-2}\tr{T}$.
Here we mean the trace of $T$ as a linear operator between the vector spaces $\cM_d$.
\end{lem}
\begin{proof}
The entanglement fidelity is
\begin{align*}
F_e(T)&=\tr{T\otimes\id\lb\ketbra{\Omega}{\Omega}\rb\ketbra{\Omega}{\Omega}}\\
&=\frac{1}{d^2}\sum\limits_{i,j,k,l=1}^d\tr{\left[T\lb\ketbra{i}{j}\rb\otimes\ketbra{i}{j}\right]\ketbra{l}{k}\otimes\ketbra{l}{k}}\\
&=\frac{1}{d^2}\sum\limits_{i,j=1}^d\tr{T\lb\ketbra{i}{j}\rb\lb\ketbra{i}{j}\rb^\dagger}.
\end{align*}
Note that $\{\ketbra{i}{j}\}_{i,j=1}^d$ is an orthonormal basis of $\cM_d$ and $\tr{T\lb\ketbra{i}{j}\rb\lb\ketbra{i}{j}\rb^\dagger}$
corresponds to the Hilbert-Schmidt scalar product between $T\lb\ketbra{i}{j}\rb$ and $\ketbra{i}{j}$. Therefore, we have that 
\begin{align*}
\sum\limits_{i,j=1}^d\tr{T\lb\ketbra{i}{j}\rb\lb\ketbra{i}{j}\rb^\dagger}=\tr{T}, 
\end{align*}
where again $\tr{T}$ is meant as the trace of $T$ as a linear operator.
\end{proof}

That is, if we know the eigenvalues or the diagonal elements of $T$ w.r.t. some basis, we may determine its entanglement and average fidelity.
The RB protocol explores the fact that twirling a channel does not change its trace and that the trace
of covariant channels has a much simpler structure, as made clear in the next corollary.
\begin{cor}\label{cor:formulatracecov}
Let $T: \cM_d \to \cM_d$ be a quantum channel that is covariant w.r.t. a unitary representation $U: G \to \bbC^d$ 
of a finite group $G$
and let $\oplus_{\alpha\in\hat{G}} \left( \bbC^{d_{\alpha}} \otimes \bbC^{m_\alpha}\right)$
be the decomposition of $\bbC^d\otimes\bbC^d$ into irreps $\alpha$ of $G$ with multiplicity $m_\alpha$ for the unitary representation $U\otimes\bar{U}$. Choose a basis s.t. 
\begin{align}\label{equ:decompcovariant1}
T=\oplus_{\alpha\in\hat{G}}  \bbI_{d_{\alpha}} \otimes B_{m_\alpha} 
\end{align}
with $B_\alpha\in\cM_{m_\alpha}$. Then
\begin{align}
F_e(T)=d^{-2}\sum\limits_{\alpha\in\hat{G}}d_\alpha\tr{B_\alpha}. 
\end{align}

\end{cor}
\begin{proof}
The claim follows immediately after we combine \cref{thm:structurecov} and \cref{lem:entangfidprojec}. 
\end{proof}

This shows that the spectrum of quantum channels that are covariant w.r.t. a unitary representation of a finite group
has much more structure and is simpler than that of general quantum channels. In particular, if the unitary representation
$U\otimes\bar{U}$ is such that $\sum_\alpha m_\alpha\ll d^2$, then we know that the spectrum of the quantum channel is 
highly degenerate and we only need to know a few points of it to estimate the trace.
We will explore this fact later in the implementation of the RB protocol.

We will now show in \cref{lem:expodecay} that the probability of measurement outcomes has a very simple form
for covariant channels and their powers.

\begin{lem}\label{lem:expodecay}
Let $T: \cM_d \to \cM_d$ be a quantum channel that is covariant w.r.t. a unitary representation 
$U: G\to\cM_d$ of a finite group $G$
and let $\oplus_{\alpha\in\hat{G}} \left( \bbC^{d_{\alpha}} \otimes \bbC^{m_\alpha}\right)$
be the decomposition of $\bbC^d\otimes\bbC^d$ into irreps $\alpha$ of $G$ with multiplicity $m_\alpha$ for the unitary representation $U\otimes\bar{U}$.
Moreover, let $\rho\in\calD_d$, $E\in\cM_d$ be a POVM element and $m\geq\max m_\alpha$.
Then there exist $\lambda_1,\ldots,\lambda_k\in\overline{B_1(0)}$, the unit ball in the complex plane, and $a_0,a_1,\ldots,a_k\in\bbC$ s.t.
\begin{align}\label{equ:decompcovariant}
\tr{T^m(\rho)E}=a_0+\sum\limits_{i=1}^ka_k\lambda_i^m.
\end{align}
Moreover, 
\begin{align}
k\leq\sum\limits_{\alpha\in\hat{G}}m_\alpha-1 
\end{align}
corresponds to the number of distinct eigenvalues of $T$ and $\lambda_i$ are its eigenvalues.
\end{lem}
\begin{proof}
As $T$ is a linear map from $\cM_d$ to $\cM_d$ it has a Jordan decomposition~\cite{Horn_2009}. That is, there exists an invertible
linear operator $X:\cM_d\to\cM_d$ such that
\begin{align*}
X^{-1}\circ T\circ X=D+N,\quad [D,N]=0.
\end{align*}
Here $D:\cM_d\to\cM_d$ is diagonal in the standard basis $\{\ketbra{i}{j}\}_{i,j=1}^d$ of $\cM_d$ with diagonal entries given
by the eigenvalues of $T$ and $N:\cM_d\to\cM_d$ nilpotent. 
As we have that $T$ is covariant, it follows from the decomposition in \cref{thm:structurecov}
that the eigenvalues can be at most $\max m_\alpha=m_0-$fold degenerate and $N^{m_0}=0$. Thus, it follows that
$T^m$ is diagonalizable, as $m\geq\max m_\alpha$. We then have
\begin{align*}
X^{-1}\circ T^m\circ{X}=D^m. 
\end{align*}
We can then rewrite the scalar product
\begin{equation*}
\tr{T^m(\rho)E}=\tr{X\circ D^m\circ X^{-1} (\rho)E}=
\tr{D^m (X^{-1}(\rho))X^\ast(E)}.
\end{equation*}
Let $b_{i,j}$ and $c_{i,j}$ be the matrix coefficient of $X^\ast(E)$ and $X^{-1}(\rho)$, respectively, in the standard basis. That is
\begin{equation*}
X^\dagger(E)=\sum\limits_{i,j=1}^db_{i,j}\ketbra{i}{j},\qquad
X^{-1}(\rho)=\sum\limits_{i,j=1}^dc_{i,j}\ketbra{i}{j}.
\end{equation*}
Exploring the fact that $D$ is diagonal in this basis we obtain 
\begin{align*}
\tr{T^m(\rho)E}= \sum\limits_{i,j=1}^db_{i,j}c_{i,j}d_{i,j}^m,
\end{align*}
where $d_{i,j}$ are just the eigenvalues of $T$, including multiplicities.
To arrive at the curve in \cref{equ:decompcovariant}, we group together all terms corresponding to the same eigenvalue $\lambda_i$.
Moreover, note that quantum channels always have $1$ in their spectrum, which gives the $a_0$ term that does not depend on $m$.
The fact that $\lambda_i\in\overline{B_1(0)}$ follows from the fact that they are given by the eigenvalues of the channel and these
are always contained in the unit circle of the complex plane~\cite{ergodicchiribella}.  
\end{proof}

Finally, we show that twirling does not change the entanglement fidelity and thus does not change the average fidelity, as
observed in~\cite{Nielsen_2002} and elsewhere in the literature.
Thus, when we want to estimate the average fidelity of a channel $T:\cM_d\to\cM_d$ we may instead work with the twirled channel
$\mathcal{T}(T)$ and explore its rich structure.
\begin{thm}\label{thm:twirlingentfid}
Let $T:\cM_d\to\cM_d$ be a quantum channel, $G$ be a finite group and $U:G\to\cM_d$ be a unitary representation.
Then
\begin{equation}
F_e(T)=F_e(\mathcal{T}(T)).
\end{equation}
\end{thm}
\begin{proof}
We present a slightly different proof of this fact here.
Note that $\cU_g^*\circ T\circ \cU_g$ is just a similarity transformation of $T$ and
thus $\tr{\cU_g^*\circ T\circ \cU_g}=\tr{T}$, where again we mean the trace of these channels as linear operators.
Thus, integrating over all $\cU_g$ does not change the entanglement fidelity, as $F_e(T)=d^{-2}\tr{T}$.
\end{proof}

%%%%%%%%%%%%%%%%%%%%%%%%%%%%%%%%%%%%%%%%%%%%%%%%%%%%%%
\section{Randomized benchmarking protocol}
\label{sec:RBprotocol}
%%%%%%%%%%%%%%%%%%%%%%%%%%%%%%%%%%%%%%%%%%%%%%%%%%%%%%
The RB protocol, as discussed in~\cite{Knill_2008, Emerson_2007, Levi_2007, Emerson_2005, Helsen_2017, Wallman_2017, Proctor_2017, Gambetta_2012, Magesan_2012, Magesan_2011, Dankert_2009} is a protocol
to estimate the average fidelity of the implementation of gates coming from some group $G$.
Its usual setting is the Clifford group, but we discuss it for general groups here.
Other papers have investigated the protocol for gates beyond Cliffords, such as~\cite{Hashagen_2018,Brown_Eastin_2018, CarignanDugas_2015}.
But all of these have restricted their analysis to some other specific group. As we will see later, we
can analyze the protocol for arbitrary groups by just investigating properties of the given unitary representation.
We mostly follow the notation of~\cite{Magesan_2011}. We assume that the error quantum channel is gate and time independent.
That is, whenever we want to implement a certain gate $\cU_g$, where $\cU_g(\cdot) = U_g \cdot U_g^\dagger$ with $U_g \in U(d)$, we actually implement 
$
\cU_g \circ T
$
for some quantum channel $T:\cM_d\to\cM_d$.
We assume that we are able to multiply and invert elements of $G$ and draw samples from the Haar measure on $G$ efficiently 
to implement this protocol, but will later relax this sampling condition.
The protocol is as follows:

\begin{description}
\item[Step 1]
Fix a positive integer $m \in \bbN$ that varies with every loop. 

\item[Step 2]
Generate a sequence of $m+1$ quantum gates. The first $m$ quantum gates $\cU_{g_1}, \ldots, \cU_{g_m}$
are independent 
and Haar distributed. 
The final quantum gate, $\cU_{g_{m+1}}$  is chosen such that in the absence of errors the net sequence is just the identity operation,
\begin{equation}
\cU_{g_{m+1}} \circ \cU_{g_{m}} \circ \ldots \circ \cU_{g_2} \circ \cU_{g_1} = \id,
\end{equation}
where $\circ$ represents composition.  Thus, the whole quantum gate sequence is 
\begin{equation}
\cS_m = \bigcirc^{m+1}_{j=1} \cU_{g_j} \circ T,
\end{equation}
where $T$ is the associated error quantum channel. 

\item[Step 3]
For every sequence, measure the sequence fidelity
\begin{equation}
\tr { \cS_m (\rho)E  },
\label{eq:AverageFidelity}
\end{equation}
where $\rho$ is the initial quantum state, including preparation errors, and $E$ is an effect operator of some POVM including measurement errors. 

\item[Step 4]
Repeat steps 2-3 and average over $M$ random realizations of the sequence of length $m$ to find the averaged sequence fidelity given by
\begin{equation}
\bar{F}(m, E, \rho)= \frac{1}{M} \sum_{m} \tr {\cS_{m} (\rho)E}.
\label{eq:fidelityaverage}
\end{equation}

\item[Step 5]
Repeat steps 1-4 for different values of $m$ and
obtain an estimate of the expected value of the sequence fidelity
\begin{equation}
F(m, E, \rho)=\tr{\mathbb{E}(\cS_m)(\rho)E}.
\end{equation}
\end{description}

%%%%%%%%%%%%%%%%%%%%%%%%%%%%%%%%%%%%%%%%%%%%%%%%%%%%%%
\subsection{Analysis of the Protocol}
\label{sec:RB}
%%%%%%%%%%%%%%%%%%%%%%%%%%%%%%%%%%%%%%%%%%%%%%%%%%%%%%
We will now show how we can estimate the average fidelity from the data produced by the protocol, that is,
an estimate on the curve $F(m, E, \rho)=\tr{\mathbb{E}(\cS_m)(\rho)E}$.
\begin{thm}\label{thm:expecttwirl}
Let $T:\cM_d\to\cM_d$ be a quantum channel and $G$ a group with a unitary representation
$U:G\to\cM_d$. If we perform the RB protocol for $G$ we have
\begin{equation}
\mathbb{E}(\cS_m)=\mathcal{T}(T)^m.
\end{equation}
\end{thm}
\begin{proof}
Although the proof is identical to the case in which $G$ is given by the Clifford group, we will cover it here for completeness.
Given some sequence $\{ \cU_{g_1},\ldots,\cU_{g_{m+1}}\}$ of unitary gates from $G$, define the unitary operators 
\begin{equation}
\label{eq:DefD}
\cD_i=\bigcirc_{j=1}^i\cU_{g_i}. 
\end{equation}
Note that we have
\begin{align}\label{equ:rephrasedi}
\nonumber
\cS_m= &\cU_{g_{m+1}} \circ T \circ \cU_{g_m} \circ T \circ \ldots \circ \cU_{g_2} \circ T \circ  \cU_{g_1}    \\\nonumber
= &\overbrace{\cU_{g_{m+1}} \circ ( \cU_{g_m} \circ \ldots \circ \cU_{g_{1}}}^{= \bbI} \circ \overbrace{\cU_{g_{1}}^\ast \circ \ldots \circ \cU_{g_m}^\ast}^{=\cD_m^\ast} ) \circ T \circ \cU_{g_m} \circ T \circ  \\
&\nonumber \ldots T \circ\underbrace{ \cU_{g_3} \circ ( \cU_{g_2} \circ \cU_{g_{1}} }_{=\cD_3}\circ\underbrace{ \cU_{g_{1}}^\ast \circ \cU_{g_2}^\ast}_{=\cD_2^\ast} ) \circ T \circ \underbrace{\cU_{g_2} \circ ( \cU_{g_{1}} }_{=\cD_2}\circ 
\underbrace{\cU_{g_{1}}^\ast}_{=\cD_1^\ast} ) \circ T \circ \underbrace{ \cU_{g_1} }_{=\cD_1}   \\
\nonumber= &\cD_m^\ast \circ T \circ \cD_m \circ \ldots \circ \cD_2^\ast \circ T \circ \cD_2 \circ \cD_1^\ast \circ T \circ \cD_1    \\
= & \bigcirc_{j=1}^m \left( \cD_j^\ast \circ T \circ \cD_j \right).
\end{align}
Here we have absorbed the first channel $T$ as SPAM error.
As we have that each of the $\cU_{g_i}$ is independent and Haar-distributed, it follows that the $\cD_i$ are independent
and Haar distributed as well. It then follows from \cref{equ:rephrasedi} that
\begin{equation*}
\mathbb{E}\lb\cS_m\rb=\mathbb{E}\lb\bigcirc_{j=1}^m \left( \cD_j^\ast \circ T \circ \cD_j \right)\rb=
\bigcirc_{j=1}^m \mathbb{E}\left( \cD_j^\ast \circ T \circ \cD_j \right)=\mathcal{T}(T)^m.
\qedhere
\end{equation*}
\end{proof}

We can then use our structural results on covariant quantum channels to obtain a more explicit form for the curve  $F(m, E, \rho)$.
\begin{cor}\label{cor:shapecurve}
Suppose we perform RB for a unitary representation $U: G\to\cM_d$ of a finite group $G$ s.t.
$U\otimes\bar{U}=\oplus_{\alpha\in\hat{G}} \left( \bbC^{d_{\alpha}} \otimes \bbC^{m_\alpha}\right)$
and a channel $T$. 
Then there exist $\lambda_1,\ldots,\lambda_k\in\overline{B_1(0)}$ and $a_0,a_1,\ldots,a_k\in\bbC$ s.t.
\begin{align}\label{equ:decompcovariant2}
F(m, E, \rho)=a_0+\sum\limits_{i=1}^ka_k\lambda_i^m.
\end{align}
for $m\geq\max m_\alpha$.
Moreover, $k\leq\sum_\alpha m_\alpha$ corresponds to the number of distinct eigenvalues of $\mathcal{T}(T)$ and $\lambda_i$ are its eigenvalues.
\end{cor}
\begin{proof}
The claim follows immediately from \cref{thm:expecttwirl} and \cref{lem:expodecay}.
\end{proof}
That is, by fitting the curve to experimental data we may obtain estimates on the $\lambda_i$
and thus on the spectrum of $\mathcal{T}(T)$. If we know the multiplicity of each eigenvalue, then we can estimate the trace as well
and thus the average fidelity.
However, in the case in which we have more than one parameter to estimate, it is not clear which eigenvalue corresponds to
which irrep and we therefore cannot simply apply the formula in \cref{cor:formulatracecov}.
Suppose we are given an estimate $\{\hat{\lambda}_1,\ldots,\hat{\lambda}_k\}$ of the parameters sorted in decreasing order and let $d_\alpha^\uparrow$ be the dimensions 
of the irreps sorted in ascending and $d_\alpha^\downarrow$ in descending order.
We define
the minimal fidelity, $F_{\min}$, to be given by
\begin{align}\label{equ:pessfid}
F_{\min}=\frac{1}{d^2}\sum d_\alpha^\downarrow\hat{\lambda}_i  
\end{align}
and the maximum fidelity, $F_{\max}$, to be given by
\begin{align}\label{equ:optfid}
F_{\max}=\frac{1}{d^2}\sum d_\alpha^\uparrow\hat{\lambda}_i.
\end{align}
That is, we look at the pairings of $d_\alpha$ and $\hat{\lambda}_i$ that produces the largest and the smallest estimate for
the fidelity. These then give the most pessimistic and most optimistic estimate, respectively.
The fact that we cannot associate a $\lambda_i$ to each irrep causes some problems in this approach from the numerical
point of view and we comment on them in \cref{sec:numericalconsid}. We also note that since the first version of this work, a modified version of the
protocol we describe here was given in~\cite{2018arXiv180602048H}. Their protocol provides a way of isolating the individual parameters in the case of irreducibly covariant channels.

%%%%%%%%%%%%%%%%%%%%%%%%%%%%%%%%%%%%%%%%%%%%%%%%%%%%%%%%
\section{Approximate Twirls}\label{sec:approxtwirls}
%%%%%%%%%%%%%%%%%%%%%%%%%%%%%%%%%%%%%%%%%%%%%%%%%%%%%%%%
In the description of our RB protocol, we assume that we are able to obtain
samples from the Haar measure of the group $G$. 
It is not possible or efficient to obtain samples of the Haar measure for most groups, but a lot of research
has been done on how to obtain approximate samples efficiently using Markov chain Monte Carlo methods, as discussed in \cref{sec:markov}.
Here we discuss how to use samples which are approximately Haar distributed for RB.
Note that these results may also be interpreted as a stability result
w.r.t. not sampling exactly from the Haar measure of $G$.
We will assume we are able to pick the $\cU_{g_k}$ independently and that they are distributed according to a measure $\nu_k$ s.t.
\begin{equation}
\|\nu_k-\mu\|_1\leq\epsilon_k,
\end{equation} 
for $\epsilon_k\geq0$. Our goal is to show that under these assumptions we may still implement the RB protocol discussed before
and obtain measurement statistics that are close to the ones obtained using Haar samples.

Motivated by this, we define the $\tilde{\nu}$-twirl of a channel.
\begin{defn}[$\tilde{\nu}$-twirl to the power $m$]
Let $\tilde{\nu}$ be a probability measure on $G^m$, $T:\cM_d\to\cM_d$
a quantum channel and $U:G\to\cM_d$ a $d-$dimensional unitary representation of $G$. 
We define the $\tilde{\nu}$-twirl to the power $m$ to be given by
\begin{equation}
\mathcal{T}_{\tilde{\nu},m}(T)=\sum\limits_{i_1,\ldots,i_m=1}^{|G|}\tilde{\nu}\lb g_{i_1},\ldots,g_{i_m}\rb\bigcirc_{k=1}^m\calU_{g_{i_{k}}}\circ T\circ\calU_{g_{i_{k}}}^*.
\end{equation}

\end{defn}
This definition boils down to the regular twirl for $\tilde{\nu}=\mu^{\otimes m}$, $\mu$ the Haar measure on $G$.
We will now show that by sampling $\cU_{g_k}$ close to Haar we have that the $\tilde{\nu}$-twirl of a channel is also close
to the usual twirl.

\begin{lem}[Approximate Twirl]\label{lem:approxtwirl}
Let $T:\cM_d\to\cM_d$ be a  quantum channel, $G$ a finite group with a $d-$dimensional unitary representation $U:G\to\cM_d$ and
$\tilde{\nu}$ a probability measure on $G^m$.
Let $\mathcal{T}_{\tilde{\nu},m}(T)$ be the $\tilde{\nu}$-twirl to the power $m$ and $\mathcal{T}(T)$ be the twirl w.r.t. 
the Haar measure on $G$ given by $\mu$.  
Moreover, let $\|\cdot\|$ be a norm s.t. $\|T\|\leq1$ for all quantum channels.
Then 
\begin{equation}
\|\mathcal{T}_{\tilde{\nu},m}(T)-\mathcal{T}(T)^m\|\leq 2\|\tilde{\nu}-\mu^{\otimes m}\|_1. 
\end{equation}
\end{lem}
\begin{proof}
Observe that we may write 
\begin{align*}
\|\mathcal{T}_{\tilde{\nu},m}(T)-\mathcal{T}(T)^m\|=\left\|\sum\limits_{i_1,\ldots,i_m=1}^{|G|}\lb\tilde{\nu}\lb g_{i_1},\ldots,g_{i_m}\rb-\frac{1}{|G|^m}\rb\bigcirc_{k=1}^m\calU_{g_{i_{k}}}\circ T\circ\calU_{g_{i_{k}}}^*\right\|. 
\end{align*}
The claim then follows from \cref{eqn:TVdistance}, as
\begin{align*}
\sum\limits_{i_1,\ldots,i_m=1}^{|G|}\left|\tilde{\nu}\lb g_{i_1},\ldots,g_{i_m}\rb-\frac{1}{|G|^m}\right|=2\|\tilde{\nu}-\mu^{\otimes m}\|_1, 
\end{align*}
the triangle inequality and the fact that $\|\bigcirc_{k=1}^m\calU_{g_{i_{k}}}\circ T\circ\calU_{g_{i_{k}}}^*\|\leq1$.
\end{proof}

Thus, in order to bound $\|\mathcal{T}_{\tilde{\nu},m}(T)-\mathcal{T}(T)^m\|$ 
in any norm in which quantum channels are contractions, it suffices to bound $\|\tilde{\nu}-\mu^{\otimes m}\|_1$.
Examples of such norms are the $1\to1$ norm and the diamond norm~\cite[Theorem 2.1]{PerezGarcia_2006}.
We remark that other notions of approximate twirling were considered in the literature~\cite{Dankert_2009,harrow2009random}, 
but these works were mostly concerned with the case of the unitary group and not arbitrary finite groups. 
Although it  would be straightforward to adapt their definitions to arbitrary finite groups, it is not clear 
at first sight that their notions of approximate twirls behave well when taking powers of channels that have been twirled approximately.
This is key for RB.
Given random unitaries $\{U_i\}_{i=1}^m$ from $G$, let $\cD_k=\bigcirc_{i=1}^k\cU_i$, as before. 

\begin{restatable}{thm}{djsameasproduct}
\label{thm:djsameasproduct}
Let $\mu$ be the Haar measure on $G$ and $\nu_1,\ldots,\nu_m$ probability measures on $G$ s.t.
\begin{equation}\label{equ:measurehaarclose1}
\|\mu-\nu_k\|_1\leq \epsilon_k, 
\end{equation}
for all $1\leq k\leq m$ and $\epsilon_k\geq0$.
Denote by $\tilde{\nu}$ the distribution of $(\cD_1,\ldots,\cD_m)$ if we pick the $\cU_k$ independently from $\nu_k$ .
Then 
\begin{equation}\label{eq:approxsamplesbound}
\|\mathcal{T}_{\tilde{\nu},m}(T)-\mathcal{T}(T)^m\|_{1\to1}\leq 4\sqrt{\frac{\log(|G|)}{1-|G|^{-1}}\sum_{k=1}^m\epsilon_k}. 
\end{equation}
\end{restatable}
\begin{proof}
We refer to~\cref{sec:proofapproxsamples} for a proof and only sketch the main steps here. 
We start by applying~\cref{lem:approxtwirl} to reduce the problem of estimating this norm to estimating the total variation distance between $\tilde{\nu}$ and $\mu$.
We then show that the total variation distance between $\tilde{\nu}$ and $\mu$ and $\otimes_{k=1}^m\mu_k$ coincide. We bound this total variation distance by the relative entropy
using Pinsker's inequality, explore tensorization properties of the relative entropy, and then use a reverse Pinsker inequality. We then obtain the final bound from~\cref{eq:approxsamplesbound}.
\end{proof}

Note that the same result holds for any norm that contracts under quantum channels, 
such as the diamond 
norm.

\begin{cor}\label{cor:POVMapprox}
Let $\mu$ be the Haar measure on $G$ and $\nu_1,\ldots,\nu_m$ probability measures on $G$ s.t.
\begin{equation}\label{equ:measurehaarclose}
\|\mu-\nu_k\|_1\leq \epsilon_k, 
\end{equation}
for all $1\leq k\leq m$ and $\epsilon_k\geq0$.
Denote by $\tilde{\nu}$ the distribution of $(\cD_1,\ldots,\cD_m)$ if we pick the $\cU_k$ independently from $\nu_k$.
Then 
\begin{equation}
|\tr{\mathcal{T}_{\tilde{\nu},m}(T)(\rho)E}-F(m,E,\rho)|\leq 4\sqrt{\frac{\log(|G|)}{1-|G|^{-1}}\sum_{k=1}^m\epsilon_k}. 
\end{equation}
\end{cor}
\begin{proof}
It follows from H\"older's  inequality  that 
\begin{align*}
|\tr{\mathcal{T}_{\tilde{\nu},m}(T)(\rho)E}-F(m,E,\rho)| 
=&|\tr{ E\lb\mathcal{T}_{\tilde{\nu},m}(T)(\rho)-\mathcal{T}\lb T\rb^m(\rho)\rb}|\\
\leq&\|E\|_\infty\|\mathcal{T}_{\tilde{\nu},m}(T)-\mathcal{T}(T)^m\|_{1\to1},
\end{align*}
where we have used the submultiplicativity of the $1\to 1$-norm.
As $E$ is the element of a POVM, we have $\|E\|_\infty\leq1$ and the claim then follows from \cref{thm:djsameasproduct}. 
\end{proof}

This shows that we may use approximate twirls instead of exact ones and obtain expectation values that are close to the perfect twirl.
Given that we want to assure that the statistics we obtain for some $m\in\bbN$ are $\delta>0$ close to our target distribution,
we would have to sample the $U_{g_k}$ such that
\begin{equation}
\|\mu-\nu_k\|_1\leq\frac{\delta^2(1-|G|^{-1})}{16\log(|G|)m}, 
\end{equation}
as can be seen by plugging in this bound in the result of \cref{cor:POVMapprox}.
If we use a random walk on a group to sample from the Haar distribution
we have to run each chain for $t_1\lb\frac{\delta^2(1-|G|^{-1})}{16\log(|G|)m}\rb$ steps, which gives
a total runtime of $\Or\lb t_{\text{mix}}\log\lb\frac{16\log(|G|)m}{\delta^2(1-|G|^{-1})}\rb\rb$.
For a fixed $\delta$, this will be efficient if the chain mixes rapidly, that is, $t_{\text{mix}}$ is small, and
we choose $m$ to be at most of the order of the dimension.

%%%%%%%%%%%%%%%%%%%%%%%%%%%%%%%%%%%%%%%%%%%%%%%%%%%%
\section{Randomized benchmarking with generators}\label{sec:Generators}
%%%%%%%%%%%%%%%%%%%%%%%%%%%%%%%%%%%%%%%%%%%%%%%%%%%
One of the downsides of the usual RB protocol~\cite{Knill_2008, Emerson_2007, Levi_2007, Emerson_2005, Wallman_2017, Proctor_2017, Gambetta_2012, Magesan_2012, Magesan_2011, Dankert_2009}
is that we assume that we may implement any gate of the group. Usually, gates have to be broken down into generators, 
as discussed in~\cite[Section 1.2.3 and Chapter 8]{kliuchnikov2014new}.
Therefore, it would be desirable both from the point of view of justifying the noise model and the implementation
level of the protocol to mostly need to implement gates from a set of generators.
We describe here a protocol to perform RB by just implementing gates from a set of generators closed
under inversion and one
arbitrary gate. We also make the additional assumption that the quantum channel that describes the noise is already approximately covariant in a sense
we will make precise soon.
This protocol is inspired by results of the last section that suggest a way of performing RB by
just implementing gates 
coming from a set $A$ that generates the group 
$G$ and is closed under inversion and one additional arbitrary gate from $G$ at each round of the protocol.
From the basic results of random walks discussed in \cref{sec:markov}, we know that if we pick gates $U_{g_1},U_{g_2},\ldots$
uniformly at random from $A$, it follows that $U_{g_b}U_{g_{b-1}}\ldots U_{g_1}$ will be approximately distributed 
like the Haar measure on $G$ for $b\simeq t_{\text{mix}}$.
However, one should note that in this setting the $\cD_i$, defined in~\cref{eq:DefD}, will not be independent of each other. To see this, note that given $\cD_i=\cU_{g}$,
we know that the distribution of the $D_{i+1}$ is restricted to elements $h\in G$ of the form $h=ag$ with $a\in A$, which clearly
show that they are not independent in general. 
However,  if we look at $\cD_{i+l}$ for $l\sim t_{\text{mix}}$, then their joint distribution
will be close to Haar.
That is, looking at $\cD_i$ and $\cD_j$ which are far enough apart from each other, we may again assume that they are both almost Haar distributed and
if we look at each $\cD_i$ individually we may assume that they are almost Haar distributed.
One way to explore this observation for RB protocols only having to implement the generators is
to look at the following class of quantum channels:
\begin{defn}[$\delta$-covariant quantum channel]\label{def:deltacov}
A quantum channel $T:\cM_d\to\cM_d$ is called $\delta$-covariant w.r.t. a unitary representation $U:G\to\cM_d$
of a group $G$, if there exist quantum channels $T_c,T_n:\cM_d\to\cM_d$ such that
\begin{equation}
T=(1-\delta)T_c+\delta T_n, 
\end{equation}
and $T_c$ is covariant w.r.t. $U$.
\end{defn}
That is, $T$ is almost covariant w.r.t. the group. 
Similar notions of approximate covariance were also introduced in~\cite{Leditzky_2018}.
Their notion of an approximate covariant channel is arguably more natural than ours, as they quantify how close a channel is to being covariant using
the minimal distance to a covariant channel in the diamond norm. 
Unfortunately, we also need information on how close the powers of the channel are to being covariant and it is not clear how to derive such bounds
using their definition but it is straightforward using ours.
Another issue related to our definition is the fact that it does not cover quantum channels that are unitary, unless they are the identity. That is because unitary channel are extremal points
of the set of quantum channels~\cite{Mendl_Wolf_2009} and thus cannot be written in any nontrivial convex combination with another quantum channel. Thus, it remains an open problem how to generalize these methods 
to unitary noise, as preliminary numerical evidence suggests that the protocol also gives good estimates in this case.
The standard example of  quantum channels that satisfy our definition are quantum channels that are close to the identity channel, i.e., we have $\delta$ small and $T_c$ the
identity channel.

We will need to fix some notation before we describe the protocol.
For a given sequence of unitaries $s_i=(U_{g_1},U_{g_2},\ldots)$ we let
$\cS_{s_i,c,d} = \bigcirc^{d}_{j=c} \cU_{g_j} \circ T$ for  $c,d\in\bbN$ and the gates chosen according to the sequence. 

Thus, if we apply random generators $b$ times as an initialization procedure and only start fitting the curve after this initialization
procedure we may also estimate the average fidelity.

This yields the following protocol.

\begin{description}
\item[Step 1]
Fix a positive integer $m \in \bbN$ that varies with every loop and another integer $b\in\bbN$.
\item[Step 2]
Generate a sequence of $b+m+1$ quantum gates, $s_i$. 
The first $b+m$ quantum gates $\cU_{g_1}, \ldots, \cU_{g_{b+m}}$ are chosen independently and uniformly at random from $A$.
The final quantum gate, $\cU_{g_{b+m+1}}$ is chosen as
\begin{equation}
\cU_{g_{b+m+1}}= (\cU_{g_{b+m}} \circ \ldots \circ \cU_{g_2} \circ \cU_{g_{1}})^{-1}.
\end{equation}
\item[Step 3]
For each sequence $s_i$, measure the sequence fidelity
\begin{equation}
\tr{ \cS_{s_i,b+1,b+m+1} (\cS_{s_i,1,b}(\rho))  E},
\label{eq:AverageFidelity2}
\end{equation}
where $\rho$ is the initial quantum state and $E$ is an effect operator of a POVM. 
\item[Step 4]
Repeat steps 2-3 and average over $M$ random realizations of the sequence of length $m$ to find the averaged sequence fidelity
\begin{equation}
\bar{F}(m, E, \rho)=\frac{1}{M}\sum\limits_{i=1}^M\tr{  \cS_{s_i,b+1,b+m+1} (\cS_{s_i,1,b}(\rho))  E}.
\end{equation}
\item[Step 5]
Repeat steps 1-4 for different values of $m$ to obtain an estimate of the expected value of the average survival probability
\begin{equation}
F(m,E,\rho) =\mathbb{E}\lb\tr{  \cS_{s_i,b+1,b+m+1} (\cS_{s_i,1,b}(\rho)) E }\rb.
\label{eq:fitmodel}
\end{equation}
\end{description}
We will now prove that this procedure gives rise to the same statistics as if we were using samples from the 
Haar distribution up to $\Or(\delta^2)$.
\begin{restatable}{thm}{onetoonegen}
\label{thm:1to1gen}
Let $T$ be $\delta-$covariant w.r.t. a unitary representation $U:G\to\cM_d$ of a finite group $G$, $A$ a subset of $G$ that generates
$G$ and is closed under inversion and $\delta>0$. Suppose
we run the protocol above with $b=t_{1}(m^{-1}\epsilon)$ for some $\epsilon$ and $m\geq b$. 
% Furthermore, let $\pi$ be the doubly-stochastic matrix associated with the random walk induced by $A$ and suppose that
% \begin{equation}\label{equ:assumpcontract}
% \|\pi^k(\nu)-\mu\|_1\leq\lambda
% \end{equation}
% for all $\nu\in\mathcal{P}(G)$ and some $\lambda\in[0,1)$.
Then
\begin{equation}
\|\cT(T)^m-\mathbb{E}(\cS_{b,b+m+1})\|_{1\to1}\leq\epsilon+\Or\lb\delta^2bm\rb. 
\end{equation}
\end{restatable}
\begin{proof}
We refer to \cref{sec:proofapproxgen} for a proof.
\end{proof}
% Determining the optimal value  of  the constant $\lambda$ is a challenging task in general~\cite[Chapter 7]{bremaud2013markov}
% Using standard methods from Markov chains it is possible to show that one can always choose the upper bound
% $\lambda=(|G|-|A|)|G|^{-1}$. To see this, observe that each element of the 

\begin{cor}\label{cor:POVMapproxgen}
Let $\cS_{b,m+b+1}$ and $b$ be as in \cref{thm:1to1gen}.
Then for any POVM element $E$ and state $\rho\in\cM_d$: 
\begin{equation}
|\tr{\mathbb{E}\lb S_{b,m+1}\rb(\rho)E}-F(m,E,\rho)|\leq \epsilon+\Or\lb\delta^2bm\rb.  
\end{equation}
\end{cor} 
\begin{proof}
The proof is essentially the same as that of \cref{cor:POVMapprox}.
\end{proof}
This shows that performing RB by only implementing the generators is feasible as long as we have a $\delta-$covariant
channel with $\delta$ small and a rapidly mixing set of generators, that is, $\delta^2bm\ll1$. 
Recall that a Markov chain is said to be rapidly mixing if the mixing time scales polylogarithmically with the size of the state space~\cite{randall2006rapidly}.
In our case, the size of the state space is given by the size of the group we are benchmarking. 
Thus, for groups whose size scales polynomially with the dimension of the system or equivalently
exponentially in the number of qubits,
this translates to $b$ scaling like $\mathcal{O}(n^k\log^k(n\epsilon^{-1}m))$, for $n$
the number of qubits and $k$ a natural number. This scaling renders the protocol reliable if $\delta$ is roughly smaller than the inverse of a polynomial on the number of qubits.

%%%%%%%%%%%%%%%%%%%%%%%%%%%%%%%%%%%%%%%%%
\section{Numerics and Examples} \label{sec:Numerics}
%%%%%%%%%%%%%%%%%%%%%%%%%%%%%%%%%%%%%%%%%
Here we show how to apply our methods to groups that might be of special interest and discuss some numerical examples.
Many relevant questions for the practical application of our work are still left open and have two different flavors: the numerical and
statistical side.
From the numerical point of view, it is not clear at first how to fit the data gathered by a RB protocol
to an exponential curve if we have several parameters. We refer to \cref{sec:numericalconsid}  for a discussion of these issues
and some proposals of how to overcome them.
From a statistical point of view, it is not clear how to derive confidence intervals for the parameters and how large we should
choose the different parameters of the protocol, such as $m$ and $M$. We refer to \cref{sec:confidence} for a discussion
of these issues and preliminary results in this direction.
\subsection{Monomial Unitary Matrices}\label{sec:monomial}
We consider how to apply our methods of generalized RB to some subgroups of the monomial unitary matrices $MU(d)$.
\begin{defn}
Let $\{\ket{i}\}_{i=1}^d$ be an orthonormal basis of $\bbC^d$. We define the group of monomial unitary matrices, $MU(d)$ to be given by
$U\in U(d)$ of the form $U=DP$ with $D,P\in U(d)$ and $D$ diagonal w.r.t. $\set{\ket{i}}_{i=1}^d$ and $P$ a permutation matrix.
\end{defn}
Subgroups of this group can be used to describe many-body states in a formalism that is broader than the stabilizer formalism of Paulis
and have other applications to quantum computation (see ~\cite{VandenNest_2011}).
As the group above is not finite and it is unreasonable to assume that we may implement diagonal gates with phases of an arbitrary 
precision, we focus on the following subgroups:
\begin{defn}
We define $MU(d,n)$ to be the 
subgroup of the monomial unitary matrices of dimension $d$ whose nonzero entries consist only
of $n-$th roots of unity.
\end{defn}
Another motivation to consider these subgroups is
that they contain the $T$-gate~\cite{Nielsen_2009},
\begin{equation}
T=\ketbra{0}{0}+e^{\rmi\frac{\pi}{4}}\ketbra{1}{1} 
\end{equation}
in case $n\geq8$. Thus these gates, together with Cliffords, constitute a universal set of quantum gates~\cite{Nielsen_2009}.
Also note that the group considered here contains the group considered in~\cite{Cross_2016}. There they also consider the group generated by diagonal matrices containing
$n-$th roots of unity, CNOTs and Pauli $X$ gates. Although the latter two are permutations, they do not generate the whole group of permutations and the groups do not coincide. 
We now show that we have to estimate two parameters for them.
\begin{restatable}[Structure of channels covariant w.r.t. monomial  unitaries]{lem}{lemmastructcov}
 \label{lem:structcov}
Let $MU(d,n)$ be such that $n\geq3$ and 
$T:\cM_d\rightarrow\cM_d$ a quantum channel. Then the following are equivalent:
\begin{enumerate} \item 
$ T(\rho)=UT\big(U^\dagger\rho U\big) U^\dagger\quad \forall U\in MU(d,n),\rho\in\cS_d$.
\item There are $\alpha,\beta\in\bbR$ so that 
\begin{equation}\label{eq:symmchannel}
T(\cdot)=\tr{\cdot}\frac{\bbI}{d}+\alpha\; \left(\id - \sum_{i=1}^d \ketbra{i}{i} \bra{i} \cdot \ket{i} \right)  +   \beta\; \left( \sum_{i=1}^d \ketbra{i}{i} \bra{i} \cdot \ket{i} - \tr{\cdot}\frac{\bbI}{d}\right).
\end{equation}
\end{enumerate}
Moreover, the terms in the r.h.s. of \cref{eq:symmchannel} are projections of rank $1$, $d^2-d$ and $d-1$, respectively.
\end{restatable}
\begin{proof}
We refer to \cref{sec:proofstruct} for a proof.
\end{proof}
This result shows that we only need to estimate two parameters when performing RB with these
subgroups. They are therefore a natural candidate to apply our methods to and we investigate this possibility further.
We begin by analyzing the complexity of multiplying and inverting elements of $MU(d,n)$.
We show this more generally for $MU(d)$, as it clearly gives an upper bound for its subgroups as well.
We may multiply and invert elements of $MU(d)$ in time $\Or(d)$.  
To multiply elements in $MU(d)$ we need to multiply two permutations of $d$ elements, which can be done in time $\Or(d)$, multiply a vector $u\in\bbC^d$ with a permutation matrix,
which can be done in time $\Or(d)$, and multiply $d$ elements of $U(1)$  with each other, which again can be done in time $\Or(d)$. This shows that multiplying elements of this group takes $\Or(d)$ operations.
To invert an element of $MU(d)$ we need to invert a permutation, which again takes $\Or(d)$, 
invert $d$ elements of $U(1)$ and apply a permutation to the resulting vector. This also takes $\Or(d)$ operations. Moreover,
one can generate a random permutation and an element of $U(1)^d$ in time $\Or(d)$, giving $\Or(Mmd)$ complexity for the classical part of the RB procedure.
Although this scaling is not efficient in the number of qubits as in the case of Clifford gates~\cite{Knill_2008}, 
the fact that it is linear in the dimension and not superquadratic as in the general case still allows for our method
to be applied to high dimensions.

To exemplify our methods, we simulate our algorithm for some dimensions and number of sequences $M$. We run the simulations
for $MU(d,8)$, as it is the smallest one that contains the $T$-gate.
We consider the case of a quantum channel $T$ that depolarizes to a random state $\sigma\in\cD_d$ with probability $(1-p)$, that is
\begin{equation}\label{equ:defdep}
T(\rho)=p\rho+(1-p)\sigma, 
\end{equation}
where $\sigma\in\cD_d$ is chosen uniformly at random from the set of states. Although the state $\sigma$ is chosen at random each time we run the protocol,
note that it is also fixed for each run. This implies that this quantum channel will in general not be covariant, as $\sigma\not=\bbI /d$ almost surely. 
It is not difficult to see that for this class of channels the entanglement fidelity is $F_e(T)=(p(d^2-1)+1)/d^2$ and we, therefore, measure
our error in terms of the parameter $p$. The results are summarized in \cref{tab:ourresultspermu}.
\begin{table}[htb]
\caption{\label{tab:ourresultspermu}Error analysis of the RB protocol described in \cref{sec:RBprotocol} to the group $MU(d,8)$ and depolarizing noise as defined in~\cref{equ:defdep}. 
We take the initial state to be $\ketbra{0}{0}$, the POVM element to be $\ketbra{0}{0}$, $p=0.9$ and we always choose $m=40$.
Moreover, we generate $100$ different channels for each combination of dimension and number of sequences. The table shows
the resulting mean and median error as well as the standard deviation for different values of $d$ and $M$. 
Here, we define the error to be given by $|F-\hat{F}|$, where $F$ is the true average fidelity of the channel and $\hat{F}$ the estimate we obtain from our protocol.
These results indicate that the protocol performs well with this range of parameters for several different dimensions, as we observed small errors for all combinations
of dimension and $M$.
Note that increasing the number of random sequences $M$ by one order of magnitude reduced the error, although this certainly requires
more experimental effort. }
\begin{indented}
\lineup
\item[]\begin{tabular}{@{}lllll}
\br
$d$ & $M$ & Mean Error & Median Error & Standard Deviation \\
~ & ~ & ($\times 10^{-3}$) &  ($\times 10^{-3}$) & ($\times 10^{-3}$) \\
\mr
\0\064&1000 & 9.17& 2.14 & 3.93\\
\0128&\0100 & 6.08& 1.48 & 2.14\\
\0128&1000 & 5.17& 1.01 & 1.13\\
1024&\0100 & 9.17& 2.14 & 3.93\\
1024&1000 & 4.55& 1.13 & 1.77\\
\br
\end{tabular}
\end{indented}
\end{table}

We also obtain numerical results for unitary noise models. Here we consider quantum channels that are given by a conjugation with a unitary $U$ of the form
\begin{align*}
U=\otimes_{j=1}^ne^{i\theta_j \sigma_{X,j}} 
\end{align*}
for systems of $n$ qubits, $\sigma_{X,j}$ the Pauli $X$ matrix acting on the $j-$th qubit and $\theta_j\in[0,2\pi)$.
We sampled channels of this form by picking the $\theta_j$ independently and uniformly at random from some interval $(0,a)$.
The magnitude of $a$ is a proxy for ``how noisy'' this unitary will be on average. Moreover, we use the methods described
in \cref{sec:isolatingparameters} to isolate the relevant parameters. The results are summarized in~\cref{tab:ourresultspermu2}.
\begin{table}[htb]
\caption{\label{tab:ourresultspermu2}Error analysis of the RB protocol described in \cref{sec:RBprotocol} to the group $MU(d,8)$ and unitary noise. 
We generate $100$ different channels for each value of $a$ and always perform the protocol for $10$ qubits.
For each run of the protocol we generate $1000$ sequences of gates and choose $m=20$.
The table shows
the resulting mean and median error as well as the standard deviation for different values of $a$. 
Here, we define the error to be given by $|F-\hat{F}|$, where $F$ is the true average fidelity of the channel and $\hat{F}$ the estimate we obtain from our protocol.
These results indicate that the protocol performs well with this range of parameters and unitary noise.
}
\begin{indented}
\lineup
\item[]\begin{tabular}{@{}lllll}
\br
$a$ & Mean Error & Median Error & Standard Deviation \\
~ & ($\times 10^{-4}$) &  ($\times 10^{-4}$) & ($\times 10^{-4}$) \\
\mr
0.1 & 3.90& 2.00 & 0.45\\
0.2&  2.63& 1.80 & 2.10\\
0.3&  3.19& 1.9 & 3.05\\
0.4&  4.11& 2.05 & 4.04\\
0.5& 4.71& 2.11 & 4.01\\
\br
\end{tabular}
\end{indented}
\end{table}
These numerical results of \cref{tab:ourresultspermu,tab:ourresultspermu2} clearly show that we may estimate the fidelity to a good degree with our procedure.
%%%%%%%%%%%%%%%%%%%%%%%%%%%%%%%%%%%%%%%%%%%%%%%%%%%%
\subsection{Clifford Group}\label{sec:cliffordgroup}
%%%%%%%%%%%%%%%%%%%%%%%%%%%%%%%%%%%%%%%%%%%%%%%%%%%%
As mentioned before, the Clifford group is the usual setup of RB, as we only
have to estimate one parameter and it is one of the main building blocks of quantum computing~\cite{Nielsen_2009}.
Thus, we apply our protocols based on approximate samples of the Haar distribution and generator based 
protocols to Clifford gates.
It is known that the Clifford group on $n$ qubits, $\mathcal{C}(n)$, is generated by the Hadamard gate $H$,
the $\pi-$gate and the $CNOT$ gate between different qubits, defined as
\begin{equation}
H = \frac{1}{\sqrt{2}}
\begin{pmatrix}
	1 & 1 \\ 1 & -1
\end{pmatrix}, \qquad
\pi = \begin{pmatrix}
	1 & 0 \\ 0 & i
\end{pmatrix}\qquad \text{ and } \qquad
CNOT = 
\begin{pmatrix}
	1 & 0 & 0 & 0 \\
	0 & 1 & 0 & 0 \\
	0 & 0 & 0 & 1 \\
	0 & 0 & 1 & 0
\end{pmatrix},
\label{eq:Gates}
\end{equation}
respectively. We refer to e.g.~\cite[Section 5.8]{gottesman1997stabilizer} for a proof
of this claim.
We need a set of generators that is closed under taking inverses for our purposes.
All but the $\pi-$gate are their own inverse, so we add the inverse of the $\pi-$gate to our set of generators to assure that
the random walk converges to the Haar measure on the Clifford group.
That is, we will consider the set $A$ of generators of the Clifford group $\mathcal{C}(n)$
consisting of Hadamard gates, $\pi-$gates and its inverse on each individual qubit and $CNOT$ between any two qubits,
\begin{align}
A=\{\pi_i,\pi^{-1}_i,H_i,CNOT_{i,j}\}.
\end{align}

To the best of our knowledge, there is no rigorous estimate available for the mixing time of the random walk
generated by $A$
and it would certainly be interesting to investigate this question further. 
However, based on our numerical results and the results of~\cite{harrow2009random}, we conjecture that it is
rapidly mixing, i.e. $t_{\text{mix}}=\Or(n^2\log(n))$.
This would be more efficient than the algorithm proposed in~\cite{koenig2014efficiently}, which takes
$\Or(n^3)$ operations.
To again test our methods we perform similar numerics as in the case of the monomial unitaries.

We simulate the following noise model:
We first pick a random isometry $V:\lb\bbC^{2}\rb^{\otimes n}\to\lb\bbC^{2}\rb^{\otimes n}\otimes\lb\bbC^{2}\rb^{\otimes n}$ and 
generate the quantum
channel 
\begin{align}\label{equ:randomchannnelmodel}
T(\rho)=p\rho+(1-p)\textrm{tr}_2\lb V\rho V^\dagger\rb, 
\end{align}
where $\textrm{tr}_2$ denotes the partial trace
over the second tensor factor.
That is, $T$ is just the convex combination of the identity and a random channel and is $\delta$-covariant w.r.t. a group
with $\delta=p$. This sampling procedure ensures that the channel $T$ will not have any further symmetries.
From the discussion in \cref{sec:Generators} we expect this to work best for $p$ close to $1$.
The results for $p$ close to $1$ are summarized in \cref{tab:ourresultsgen}.
\begin{table}[htb]
\caption{\label{tab:ourresultsgen}For each combination of $p,M$ and $b$ we generate $20$ different random quantum channels and perform
    generator RB for the Clifford group on $5$ qubits. In all these cases we pick $m=20$.
    The average error
    is defined as the average of the absolute value between the exact fidelity and the one estimated using our protocol.
    The table shows the average error and its standard deviation in terms of different choices of $b$, $M$ and $p$.}
		\begin{indented}
		\lineup
       \item[]\begin{tabular}{@{}lllll}
        \br
        $p$ & $b$& $M$ & Average Error & Standard Deviation of Error \\
				~ & ~& ~ & ($\times 10^{-3})$ & ($\times 10^{-4})$ \\
        \mr
        0.98 & 10 & \010 & 5.49 & 1.38 \\
        0.95 & 10 & 100 & 1.44 & 3.92 \\
        0.95 & \05 & 100 & 1.52 & 7.94 \\
        0.95 & \05 & \020 & 1.56 & 7.44 \\
        0.90 & 10 & \020 & 3.20 & 1.58 \\
        0.80 & 10 & \050 & 8.63 & 6.01 \\
        \br
    \end{tabular}
    \end{indented}
\end{table}
The average error increases as the channel becomes noisier, but generally speaking we are able to obtain an
estimate which is $10^{-3}$ close to the true value with $M$ around $20$ and $m=20$.

We also performed some numerical experiments for $p$ significantly away from $1$, which are summarized in \cref{tab:ourresultsgenbad}.
\begin{table}[htb]
\caption{\label{tab:ourresultsgenbad}For each combination of $p,M$ and $b$ we generate $20$ different random quantum channels and perform
    generator RB for the Clifford group on $5$ qubits. In all these cases we pick $m=20$.
    The average error
    is defined as the average of the absolute value between the exact fidelity and the one estimated using our protocol.
     The table shows the average error and its standard deviation in terms of different choices of $b$, $M$ and $p$.}
		\begin{indented}
		\lineup
    \item[]\begin{tabular}{@{}lllll}
        \br
        $p$ & $b$ & $M$ & Average Error & Standard Deviation of Error \\
				~ & ~ & ~ & ($\times 10^{-2})$ & ($\times 10^{-3})$ \\
        \mr
        0.7 & 5 & 100 & \02.07 & \01.15 \\
        0.65 & 5 & 100 & \02.29 & \01.95 \\
        0.60 & 5 & 100 & 27.1 & 52.30 \\
        0.55 & 5 & 100 & 44.5 & 67.30 \\
        \br
    \end{tabular}
    \end{indented}     
\end{table}

The noise model above favors quantum channels with a high Kraus rank. Here we also consider the case of quantum channels of the form
\begin{align*}
T(\rho)=p\rho+(1-p)U\rho U^{\dagger}, 
\end{align*}
where $U$ is a randomly chosen (Haar) unitary. These channels have Kraus rank $2$ and are $\delta$-covariant with $\delta=p$. The numerical results can be found in~\cref{tab:ourresultsgenunit}.
\begin{table}[htb]
\caption{\label{tab:ourresultsgenunit}For each combination of $p,M$ and $b$ we generate $20$ different convex combinations of the identity and a random unitary and perform
    generator RB for the Clifford group on $5$ qubits. In all these cases we pick $m=20$.
    The average error
    is defined as the average of the absolute value between the exact fidelity and the one estimated using our protocol.
    The table shows the average error and its standard deviation in terms of different choices of $b$, $M$ and $p$.}
		\begin{indented}
		\lineup
       \item[]\begin{tabular}{@{}lllll}
        \br
        $p$ & $b$& $M$ & Average Error & Standard Deviation of Error \\
				~ & ~& ~ & ($\times 10^{-3})$ & ($\times 10^{-4})$ \\
        \mr
        0.98 & 10 & 100 & 2.30 & 9.44 \\
        0.95 & 10 & 100 & 1.15 & 9.19 \\
	0.90 & 10 & 100 & 3.62 & 2.22 \\
        0.85 & 10 & 100 & 6.67 & 39.4 \\
        0.80 & 10 & 100 & 83.4 & 55.9 \\
        \br
    \end{tabular}
    \end{indented}
\end{table}

These results show that these methods are effective to estimate the average fidelity under less restrictive assumptions
on the gates we may implement
using RB if we have a high fidelity, as indicated in \cref{tab:ourresultsgen,tab:ourresultsgenunit}.
However, in case we do not have a high fidelity, these methods are not reliable, as can be seen in \cref{tab:ourresultsgenbad}.
Note that our numerical results seem to indicate that 
the cut-off of the range of average fidelities we can reliably detect occurs at larger values of the fidelity in the case of channels with lower Kraus rank,
as can be see in \cref{tab:ourresultsgenunit}.
This should not severely restrict the applicability of these methods, as one is usually interested in the high fidelity regime
when performing RB.

\begin{figure}
(a)
\begin{center}
\includegraphics[trim={1.5cm 7cm 0.5cm 7cm},clip,width=0.85\textwidth]{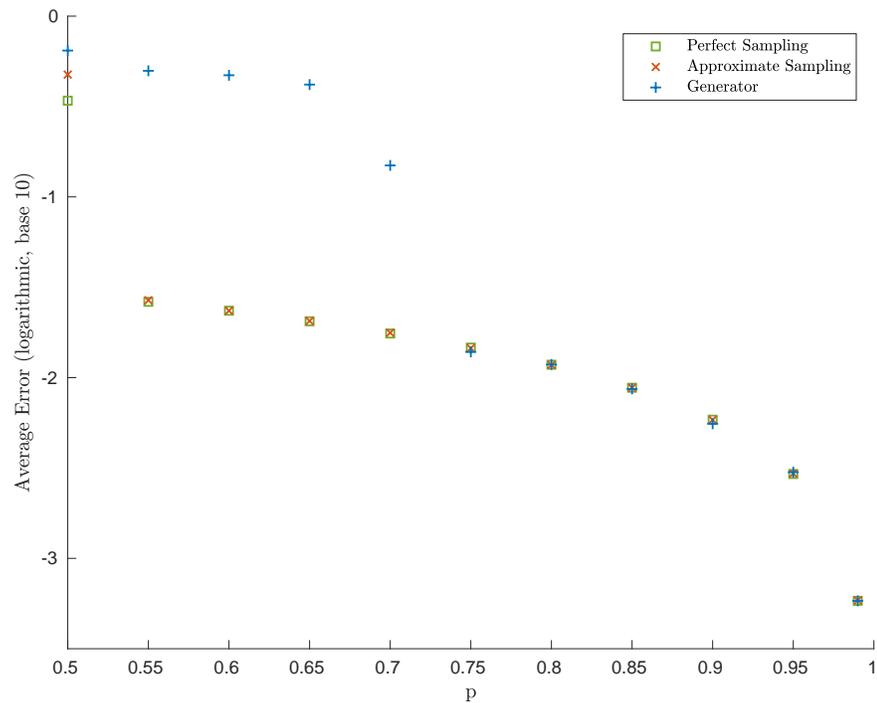}     
\end{center}
(b)
\begin{center}
\includegraphics[trim={1.5cm 7cm 0.5cm 7cm},clip,width=0.85\textwidth]{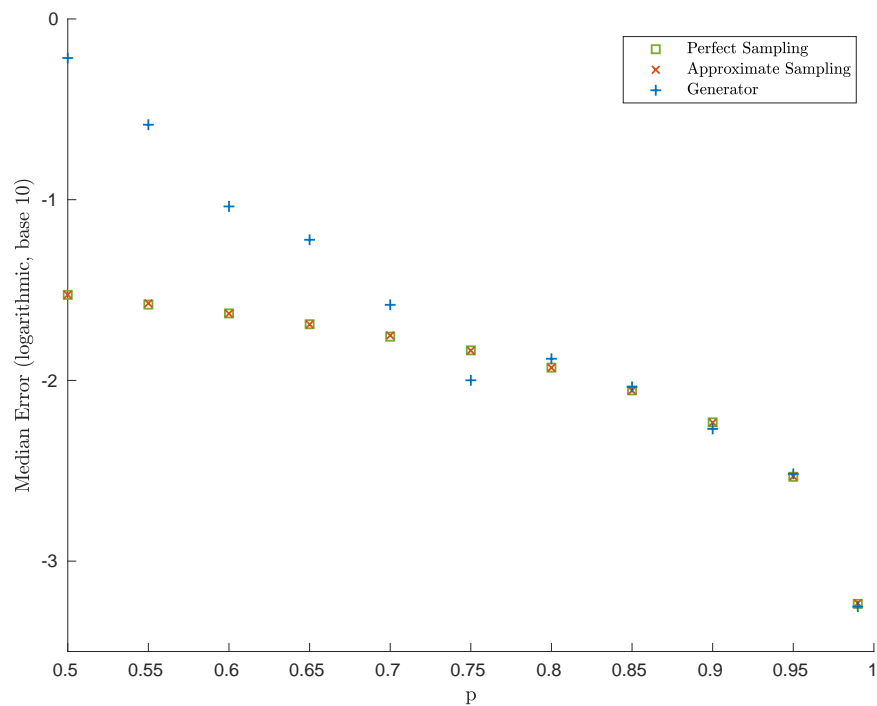}
\end{center}
\caption{\label{fig:comparison}Plot of the average error (a) and mean error (b) as a function of $p$ for different
versions of the RB protocol for the Clifford group and the random quantum channel noise model as defined in~\cref{equ:randomchannnelmodel}. For each value of $p$ we generated $20$ instances of the random channel with $M=100$ and $m=20$.
For the generator RB we chose
$b=5$ and to obtain the approximate samples we ran the chain for $20$ steps.}
\end{figure}

Finally, in \cref{fig:comparison} we compare the three different RB protocols discussed in this paper.
We compare the usual RB protocol, which we call the perfect sampling protocol, to the one with approximate samples
and the generator RB for the random quantum channel noise model. The curve makes clear that using approximate and exact samples leads to virtually indistinguishable estimates
and that all protocols have similar performance for $p$ close to $1$.
%%%%%%%%%%%%%%%%%%%%%%%%%%%%%%%%%%%%%%%%%%%%%%%%%%%%%%
\section{Conclusion and Open Problems}\label{sec:Conclusion}
%%%%%%%%%%%%%%%%%%%%%%%%%%%%%%%%%%%%%%%%%%%%%%%%%%%%%%
We have generalized the RB protocol to estimate the average gate fidelity
of unitary representations of arbitrary finite groups.
Our protocol is efficient when multiplying, inverting and sampling elements from the group can be done efficiently and we have shown
some potential applications that go beyond the usual Clifford one. Moreover, we showed that using approximate samples instead of perfect ones 
from the Haar measure on the group does not lead to great errors. This can be seen as a stability result for RB protocols 
w.r.t. sampling which was not available in the literature and is also relevant in the Clifford case.
We hope that this result can be useful in practice when one is not given a full description of the group but rather a set of generators.
Moreover, we have shown how to perform RB by just implementing a set of generators and one arbitrary gate under some
noise models. This protocol could potentially be more feasible for applications, as the set of gates we need to implement is on average
simpler.

However, some questions remain open and require further work. 
It is straightforward to generalize the technique of interleaved RB to this more general scenario and this would also be a
relevant development.
It would be important to derive confidence intervals for the estimates
as was done for the Clifford case in~\cite{Wallman_Flammia_2014,Helsen_2017}. Moreover, it would be relevant to estimate not only the mean fidelity but also the variance of this quantity. 
The assumption that the noisy channel is the same for all gates is not realistic in many scenarios and should be seen as a $0$-order approximation, as in~\cite{Magesan_2012}. 
It would be desirable to generalize our results to the case in which the channel depends weakly on the gate.

%%%%%%%%%%%%%%%%%%%%%%%%%%%%%%%%%%%%%%%%%%%%%%%%%%%%%%
\ack
%%%%%%%%%%%%%%%%%%%%%%%%%%%%%%%%%%%%%%%%%%%%%%%%%%%%%%
We would like to thank the anonymous for various valuable comments.
D.S.F. acknowledges support from the graduate program TopMath of the Elite Network of Bavaria, the 
TopMath Graduate Center of TUM Graduate School at Technische Universit\"{a}t M\"{u}nchen
and by the Technische Universit\"at M\"unchen – Institute for Advanced Study,
funded by the German Excellence Initiative and the European Union Seventh Framework
Programme under grant agreement no. 291763.

A.K.H.'s work is supported by the Elite Network of Bavaria through the PhD programme
of excellence \textit{Exploring Quantum Matter}.
\newpage
\appendix
\section*{Appendix}
\setcounter{section}{0}
%%%%%%%%%%%%%%%%%%%%%%%%%%%%%%%%%%%%%%%%%%%%%%%%%%%%%%
\section{Numerical Considerations}\label{sec:numericalconsid}
%%%%%%%%%%%%%%%%%%%%%%%%%%%%%%%%%%%%%%%%%%%%%%%%%%%%%%
Here we gather some comments on the numerical issues associated with the RB procedure
when estimating more than one parameter.

%%%%%%%%%%%%%%%%%%%%%%%%%%%%%%%%%%%%%%%%%%%%%%%%%%%%%%
\subsection{Fitting the Data to Several Parameters}\label{sec:fitting}
%%%%%%%%%%%%%%%%%%%%%%%%%%%%%%%%%%%%%%%%%%%%%%%%%%%%%%
In order to be able to estimate the average fidelity following the protocols discussed so far, it is
necessary to fit noisy data points $\{x_i\}_{i=1}^m\subset\bbR$ to a curve $f:\bbR\to\bbR$ of the form
\begin{equation}
f(x)=a_0+\sum\limits_{k=1}^na_k\rme^{-b_k x}, 
\end{equation}
with $a_0,a_1,\ldots,a_n,b_1,\ldots,b_n\in\bbC$.
Although this may look like an innocent problem at first sight, fitting noisy data to exponential curves
is a difficult problem from the numerical point of view for large $n$. It suffers from many stability 
issues, as thoroughly discussed in~\cite{hokanson2013numerically}. 
Here we are going to briefly comment on some of the issues and challenges faced when
trying to fit the data, although we admittedly only scratch the surface. For a more thorough analysis of some methods and issues,
we refer to~\cite{hokanson2013numerically,Holmstr_m_2002}.

We assume that we know the maximum number of different parameters, $2n+1$, which
we are fitting. This is given by the structure of the unitary representation at hand, as discussed in~\cref{lem:expodecay}.
Luckily, significant progress has been made in the recent years to develop algorithms to overcome the issues faced 
in this setting and it is now possible to fit curves to data with a moderate number of parameters.
It is also noteworthy that for $n=2$ there exist stable algorithms based on geometric sums~\cite{Holmstr_m_2002} which works
for equispaced data, as is our case.
For estimating more than two parameters one can use the algorithms proposed in~\cite{hokanson2013numerically},
available at~\cite{expofitcode}. It should be said that the reliability and convergence of most algorithms
found in the literature depends strongly on the choice of a good initial point. This tends not to be a problem, as
we might have some assumptions where our fidelity approximately lies and choose the initial $b_k$ accordingly.
What could be another source of numerical instabilities is the fact that we have to input the model with a number
of parameters, $n$.
In case the eigenvalues of $T$ are very close for different irreps, then this will lead to numerical instabilities.
This is the case if the noise is described by a depolarizing channel, for example. Furthermore, it might be the case that 
the initial state in our protocol does not intersect with all eigenspaces of the channel. This may lead to
some parameters $a_k$ being $0$ and we are not able to estimate some of the $b_k$ from them.

Moreover, it is in principle not possible to tell which parameter corresponds to which irrep given the decomposition 
in \cref{lem:expodecay}, which is again necessary to estimate the trace of the channel. So even in the case in which we have 
a small number of parameters, it is important that the different irreps associated to our parameters have a similar dimension
or to assume that the spectrum of the twirled channel contains eigenvalues that are very close to each other.
In this way, the most pessimistic estimate on the fidelity, as defined in~\cref{equ:pessfid}, is not very far from the most 
optimistic, defined in \cref{equ:optfid}.
This is one of the reasons we focus on examples that only have a small number of parameters, say $1$ or $2$,
and irreps of a similar dimension to avoid having numerical instabilities or estimates that range over an interval that is too large.

It is therefore important to develop better schemes to fit the data in the context of RB for more than one or
two parameters. This is important from a statistical point of view, as it would be desirable to obtain confidence intervals
for the parameters from the RB data. We will further develop this issue in \cref{sec:confidence}.
It would be worthwhile pursuing a Bayesian approach to this problem, 
as was done in~\cite{ferriebayesian} for the usual RB protocol.
%%%%%%%%%%%%%%%%%%%%%%%%%%%%%%%%%%%%%%%%%%%%%%%%%%%%%%
\subsection{Isolating the parameters}\label{sec:isolatingparameters}
One way to possibly deal with this issue is to isolate each parameter, that is, by preparing states that only have support on one
of the irreps that are not the trivial one.
In the case of non-degenerate unitary representations, discussed in \cref{thm:structirredcov}, we have the following:
\begin{thm}[Isolating parameters]\label{thm:isolatingparameters}
Let $U:G\to\cM_d$ be a simply covariant irrep of a finite group $G$ and $T:\cM_d\to\cM_d$ a channel which is covariant
w.r.t. $U$.
Then, for all eigenvalues $\lambda_\alpha$ there is a quantum state $\rho_\alpha=\frac{\bbI}{d}+X$, where $X=X^\dagger$ and $\tr{X}=0$,
such that
\begin{align}\label{equ:almosteigenvaluechannel}
T^m\lb\rho_\alpha\rb=\frac{\bbI}{d}+\lambda_\alpha^m X. 
\end{align}
\end{thm}
\begin{proof}
Consider the projections to the irreducible subspaces $P_\alpha$ defined in \cref{def:projec}.
For a self-adjoint operator $X\in\cM_d$ we have that
\begin{align*}
P_\alpha(X)^\dagger=\frac{\chi^\alpha(e)}{|G|} \sum_{g \in G} \chi^\alpha \left( g\right) U_g^\dagger X U_g=P_\alpha(X),
\end{align*}
as we are summing over the whole group and $\overline{\chi^\alpha \left( g^{-1}\right)}=\chi^\alpha \left( g\right)$.
Therefore, we have that the $P_\alpha$ are hermiticity preserving. As the image of $P_\alpha$ is the eigenspace corresponding
to the irreps, we thus only have to show that there exists a self-adjoint $X$ such that $P_\alpha(X)\not=0$.
But the existence of such an $X$ is clear, as we may choose a basis of $\cM_d$ that consists of self-adjoint operators.
Moreover, as for $\alpha$ not the trivial representation all eigenvectors are orthogonal to $\bbI$, it follows that
$\tr{X}=0$ and that for $\epsilon>0$ small enough $\frac{\bbI}{d}+\epsilon X$ is positive semidefinite.
To finish the proof, note that simply irreducible channels always satisfy
\begin{equation*}
T(\bbI)=\bbI. 
\qedhere
\end{equation*}
\end{proof}
Note that this also proves that the spectrum of irreducibly covariant channels is always real.
That is, if we can prepare a state such as in \cref{equ:almosteigenvaluechannel}, then we can perform the RB
with this as an initial state and estimate the eigenvalue corresponding to each irrep.
This would bypass the problems discussed in \cref{sec:fitting}.
The proof of \cref{thm:isolatingparameters} already hints a way of determining how to isolate the parameter: just apply 
the projector $P_\alpha$ to some states $\rho_i$. If the output is not $0$, then we can in principle write down a state
that ``isolates'' the parameter as in the proof of \cref{thm:isolatingparameters}.
This idea was explored in~\cite{2018arXiv180602048H} to obtain a way of isolating the parameters for irreducibly covariant channels
and general representation under additional assumptions.
However, in the case of the monomial unitary matrices discussed in \cref{sec:monomial}, we can examine the projections and see how to
isolate the parameters.
To isolate the parameter $\alpha$ in \cref{eq:symmchannel}, we can prepare initial states $\rho\in\cD_d$ that have $1/d$
as their diagonal elements and at least one nonzero off-diagonal element, as then the projector corresponding to $\beta$ 
vanishes on $\rho$ and does not vanish on the one corresponding to $\alpha$.
To isolate the parameter $\beta$, one can prepare states $\rho$ that are diagonal in the computational basis but are not the maximally
mixed state, as can be seen by direct inspection.
%%%%%%%%%%%%%%%%%%%%%%%%%%%%%%%%%%%%%%%%%%%%%%%%%%%%%%
\section{Statistical Considerations}\label{sec:confidence}
%%%%%%%%%%%%%%%%%%%%%%%%%%%%%%%%%%%%%%%%%%%%%%%%%%%%%%
One of the main open questions left in our work is how to derive good confidence intervals 
for the average fidelity. 
For the case of the Clifford group, discussed in~\cref{sec:cliffordgroup}, 
one can directly apply the results of~\cite{Wallman_Flammia_2014,Helsen_2017},
but it is not clear how one should pick $m$ and $M$ for arbitrary finite groups.
Especially in the case in which we are not working with Cliffords, it is not clear how
many sequences per point, $M$, we should gather and how big $m$ should be, as it depends on the choice 
of the algorithm picked for fitting the curve. As noted in \cref{sec:fitting}, this is not a trivial
problem from a numerical point of view.
However, it is possible to obtain estimates on how much the observed survival probability deviates from 
its expectation value by just using Hoeffding's inequality:
\begin{thm}
Let $\bar{F}(m,E,\rho)$ be the observed average fidelity with $M$ sequences and
${F}(m,E,\rho)$ the average fidelity for any of the protocols discussed before and $\epsilon>0$.
Then:
\begin{align}
\bbP(|F(m,E,\rho)-\bar{F}(m,E,\rho)|\geq\epsilon)\leq \rme^{-2M\epsilon^{2}}.
\end{align}
\end{thm}
\begin{proof}
This is just a straightforward application of Hoeffding's inequality~\cite{hoeffding}, 
as $\bar{F}(m,E,\rho)$ is just the empirical average of a random variable whose value is contained in $[0,1]$
and whose expectation value is ${F}(m,E,\rho)$.
\end{proof}
This bound is extremely general, as we did not even have to use any property of the random variables or of the group
at hand. One should not expect it to perform well for specific cases and the scaling it gives is still undesirable for applications.
Indeed, to assure we are $10^{-4}$ close to the expectation value with probability of $0.95$, we need around $6\times 10^8$ sequences, which
is not feasible. Thus, it is necessary to derive more refined bounds for specific groups.

%%%%%%%%%%%%%%%%%%%%%%%%%%%%%%%%%%%%%%%%%%%%%%%%%%%%%%
\section{Proof of Theorem \ref{thm:djsameasproduct}}\label{sec:proofapproxsamples}
%%%%%%%%%%%%%%%%%%%%%%%%%%%%%%%%%%%%%%%%%%%%%%%%%%%%%%

\djsameasproduct*

\begin{proof}
From \cref{lem:approxtwirl} it suffices to show 
\begin{align*}
\|\tilde{\nu}-\mu^{\otimes m}\|_1\leq2\sqrt{\frac{\log(|G|)}{1-|G|^{-1}}\sum_{k=1}^m\epsilon_k},  
\end{align*}
as the $1\to1$ norm contracts under quantum channels~\cite{PerezGarcia_2006}.

We will first show that 
\begin{align*}
\|\tilde{\nu}-\mu^{\otimes m}\|_1=\|\otimes_{k=1}^m\nu_k-\mu^{\otimes m}\|_1. 
\end{align*}
We may rewrite the distribution $\tilde{\nu}$ in terms of the $\nu_k$ as follows:
\begin{align*}
\bbP(\cD_1=g_1,\cD_2=g_2,\ldots,\cD_m=g_m)
=&\bbP(U_1=g_1,U_2=g_2g_1^{-1},\ldots,U_m=g_mg_{m-1}^{-1})\\ 
=&\nu_1(g_1)\nu_2(g_2g_1^{-1})\ldots\nu_m(g_mg_{m-1}^{-1}), 
\end{align*}
as the $U_{g_i}$ are independent. 

Note that the map $\sigma:G^m\to G^m$,
$\lb g_1,\ldots,g_m\rb\mapsto\lb g_1,g_2g_1^{-1},\ldots,g_mg_{m-1}^{-1}\rb$ is bijective.
Moreover, we have
$\tilde{\nu}=\otimes_{k=1}^m\nu_k\circ \sigma$. As the total variation norm is invariant under compositions
with bijections on the state space, we have
\begin{equation*}
\|\tilde{\nu}-\mu^{\otimes m}\|_1=\|\otimes_{k=1}^m\nu_k\circ \sigma-\mu^{\otimes m}\|_1
=\|\otimes_{k=1}^m\nu_k-\mu^{\otimes m}\circ \sigma^{-1}\|_1=\|\otimes_{k=1}^m\nu_k-\mu^{\otimes m}\|_1,
\end{equation*}
where the last equality follows from the fact that the Haar measure is invariant under bijections.
We will now bound $\|\otimes_{k=1}^m\nu_k-\mu^{\otimes m}\|_1$.
By Pinsker's inequality~\cite{Sason_2015}, we have 
\begin{align}\label{equ:channelentropy}
\|\otimes_{k=1}^m\nu_k-\mu^{\otimes m}\|_1^2\leq 4\Dkl{\otimes_{k=1}^m\nu_k}{\mu^{\otimes m}}=4\sum\limits_{k=1}^m\Dkl{\nu_k}{\mu}.
\end{align}
Here $D$ is the relative entropy.
In~\cite[Theorem 1]{Sason_2015} they show that 
\begin{align*}
\Dkl{\nu_k}{\mu}\leq \frac{\log(|G|)}{1-|G|^{-1}}\|\mu-\nu_k\|_1 
\end{align*}
and from \ref{equ:measurehaarclose1} it follows that 
\begin{align}\label{equ:finalboundrelent}
\Dkl{\nu_k}{\mu}\leq \frac{\log(|G|)}{1-|G|^{-1}}\epsilon_k. 
\end{align}
Combining \cref{equ:finalboundrelent} with \cref{equ:channelentropy} and taking the square root yields the claim. 
\end{proof}
%%%%%%%%%%%%%%%%%%%%%%%%%%%%%%%%%%%%%%%%%%%%%%%%%%%%%%%%%%%%%%%%%%%%%%
\section{Proof of Theorem \ref{thm:1to1gen}}\label{sec:proofapproxgen}
%%%%%%%%%%%%%%%%%%%%%%%%%%%%%%%%%%%%%%%%%%%%%%%%%%%%%%

\onetoonegen*

\begin{proof}
Let $T_c$ and $T_n$ be as in \cref{def:deltacov}. Then we have
\begin{align*}
\cT(T)=(1-\delta)T_c+\delta \cT(T_n),
\end{align*}
as $T_c$ is already covariant, and 
\begin{align}\label{equ:realuptosecond}
\nonumber
\cT(T)^m=&(1-\delta)^mT_c^m+\delta(1-\delta)^{m-1}\sum_{j=0}^{m-1}T_c^{j}\cT(T_n)T_c^{m-j-1}\\
&  +\delta^2(1-\delta)^{m-2}\sum\limits_{j_1+j_2+j_3=m-2}T_c^{j_1}\cT(T_n)T_c^{j_2}\cT(T_n)T_c^{j_3}+\Or(\delta^3).
\end{align}
Moreover, as $T_c$ is covariant w.r.t. this unitary representation, we have
\begin{multline}
\mathbb{E}(\cS_{b,m+b+1})=(1-\delta^m)T_c+\delta(1-\delta)^{m-1}\sum_{j=0}^{m-1}\mathbb{E}\lb T_c^{j}\cD_{m-j}T_n\cD_{m-j}^*T_c^{m-j-1}\rb\\
 +\delta^2(1-\delta)^{m-2}\sum\limits_{j_1+j_2+j_3=m-2}\mathbb{E}\lb T_c^{j_1}\cD_{j_2+1}T_n\cD_{j_2+1}^*T_c^{j_2}\cD_{j_3+1}T_n\cD_{j_3+1}^*T_c^{j_3}\rb
+\Or(\delta^3)
\label{equ:approxuptosecond}
\end{multline}
It is clear that the terms of $0-$order in $\delta$ in \cref{equ:realuptosecond} and \cref{equ:approxuptosecond}
coincide.
Comparing each of the summands of first order we obtain:
\begin{align*}
&\mathbb{E}\lb T_c^{j}\cD_{m-j}T_n\cD_{m-j}^*T_c^{m-j-1}\rb-T_c^{j}\cT(T_n)T_c^{m-j-1}\\
=&\sum\limits_{g\in G}\lb\nu_{m-j}(g)-\frac{1}{|G|}\rb T_c^{j}\cU_{g}T_n\cU_{g}^*T_c^{m-j-1},
\end{align*}
where $\nu_{m-j}$ is the distribution of $\cD_{m-j}$.
Comparing the terms of second order we obtain:
\begin{align*}
&\mathbb{E}\lb T_c^{j_1}\cD_{j_2+1}T_n\cD_{j_2+1}^*T_c^{j_2}\cD_{j_3+1}T_n\cD_{j_3+1}^*T_c^{j_3}\rb-T_c^{j_1}\cT(T_n)T_c^{j_2}\cT(T_n)T_c^{j_3}\\
=&\sum\limits_{g_1,g_2\in G}\lb\tau_{j_3+1,j_2+1}(g_1,g_2)-\frac{1}{|G|^2}\rb T_c^{j_1}\cU_{g_1}T_n\cU_{g_1}^*T_c^{j_2}\cU_{g_2}T_n\cU_{g_2}^*T_c^{j_3}.
\end{align*}
Here $\tau_{j_3+1,j_2+1}$ is the joint distribution of $\cD_{j_3+1}$ and $\cD_{j_2+1}$.
Then, using arguments similar to those of \cref{thm:djsameasproduct}, we have that
\begin{multline*}
\nonumber
\|\cT(T)^m-\mathbb{E}(S_{b,m+1})\|_{1\to1} \\ 
\leq\delta(1-\delta)^{m-1}\sum\limits_{j=1}^m\|\nu_{j}-\mu\|+
\delta^2(1-\delta)^{m-2}\sum\limits_{j_1=1}^{m-1}\sum\limits_{j_2=j_1+1}^{m}\|\tau_{j_1,j_2}-\mu^{\otimes 2}\|_1+\Or(\delta^3).
\end{multline*}
Now, from our choice of $b$, we have $\|\nu_{j}-\mu\|_1\leq\frac{\epsilon}{m}$.
Furthermore, we have that 
\begin{align*}
\tau_{j_1,j_2}(g_1,g_2)=\bbP(\cD_{j_1}=\cU_{g_1},\cD_{j_2}=\cU_{g_2})=\bbP(\cD_{j_2}=\cU_{g_2}|\cD_{j_1}=\cU_{g_1})\bbP(\cD_{j_1}=\cU_{g_1}). 
\end{align*}
By the construction of the $\cD_j$, it holds that 
\begin{align*}
\bbP(\cD_{j_2}=\cU_{g_2}|\cD_{j_1}=\cU_{g_1})=\pi^{j_2-j_1}(g_1,g_2), 
\end{align*}
where $\pi$ is the stochastic matrix of the chain generated by $A$.
From this we obtain
\begin{align}
\nonumber
&\sum\limits_{g_1,g_2\in G}\left|\tau_{j_1,j_2}(g_1,g_2)-\frac{1}{|G|^2}\right|=\sum\limits_{g_1,g_2\in G}\left|\nu_{j_1}(g_1)\pi^{j_2-j_1}(g_1,g_2)-\frac{1}{|G|^2}\right|\\
\leq&\sum\limits_{g_1,g_2\in G}\left|\nu_{j_1}(g_1)-\frac{1}{|G|}\right|\pi^{j_2-j_1}(g_1,g_2)+\left|\frac{1}{|G|}\pi^{j_2-j_1}(g_1,g_2)-\frac{1}{|G|^2}\right|.
\label{equ:secondorderstuff}
\end{align}
As the matrix $\pi$ is doubly stochastic, summing over $g_2$ first
\begin{align*}
\sum\limits_{g_1,g_2\in G}\left|\nu_{j_1}(g_1)-\frac{1}{|G|}\right|\pi^{j_2-j_1}(g_1,g_2)=\sum\limits_{g_1\in G}\left|\nu_{j_1}(g_1)-\frac{1}{|G|}\right|\leq\epsilon m^{-1},
\end{align*}
which again follows from our choice of $b$.
We now estimate the other term in \cref{equ:secondorderstuff},
\begin{equation}\label{equ:sumaverage}
\sum\limits_{j_1=1}^{m-1}\sum\limits_{j_2=j_1+1}^m\sum\limits_{g_1,g_2\in G}\frac{1}{|G|}\left|\pi^{j_2-j_1}(g_1,g_2)-\frac{1}{|G|}\right|
\end{equation}
Note that for a fixed $g_1$, 
\begin{align*}
\sum\limits_{g_2\in G}\left|\pi^{j_2-j_1}(g_1,g_2)-\frac{1}{|G|}\right|
\end{align*}
is just the total variation distance between the Markov chain starting at $g_1$ and the Haar measure after $j_2-j_1$ steps.
Thus, in case $j_2-j_1\geq t_1\left(\epsilon m^{-1}\right)$,
\begin{equation}\label{equ:sumaverage2}
\sum\limits_{g_1,g_2\in G}\frac{1}{|G|}\left|\pi^{j_2-j_1}(g_1,g_2)-\frac{1}{|G|}\right|\leq\frac{\epsilon}{m}.
\end{equation}
and in case $j_2-j_1\leq t_1\left(\epsilon m^{-1}\right)$ we have the trivial estimate
\begin{equation}\label{equ:sumaverage3}
\sum\limits_{g_1,g_2\in G}\frac{1}{|G|}\left|\pi^{j_2-j_1}(g_1,g_2)-\frac{1}{|G|}\right|\leq2. 
\end{equation}
Combining inequalities~\cref{equ:sumaverage2} and~\cref{equ:sumaverage3}, we obtain
\begin{align*}
\sum\limits_{j_1=1}^{m-1}\sum\limits_{j_2=j_1+1}^m\sum\limits_{g_1,g_2\in G}\frac{1}{|G|}\left|\pi^{j_2-j_1}(g_1,g_2)-\frac{1}{|G|}\right|=\mathcal{O}(mt_1(\epsilon m^{-1})).
\end{align*}

Putting all inequalities together, we obtain the claim.
\end{proof}

%%%%%%%%%%%%%%%%%%%%%%%%%%%%%%%%%%%%%%%%%%%%%%%%%%%%%%%%%%%%%%%%%%%
\section{Proof of Lemma \ref{lem:structcov}}\label{sec:proofstruct}
%%%%%%%%%%%%%%%%%%%%%%%%%%%%%%%%%%%%%%%%%%%%%%%%%%%%%%

\lemmastructcov*

\begin{proof} (2) $\Rightarrow$ (1) can be seen by direct inspection. 
In order to prove the converse, we consider the Choi-Jamiolkowski state
$\tau_T:=\frac1d \sum_{i,j=1}^d T\big(\ketbra{i}{j}\big)\otimes \ketbra{i}{j}$. 
Then (1) is equivalent to the statement that $\tau_T$ commutes with all unitaries 
of the form $U\otimes\bar{U}$, $U\in MU(d,n)$. That is, we have
\begin{align*}
\sum_{i,j=1}^dU\otimes\bar{U}\lb T\big(\ketbra{i}{j}\big)\otimes \ketbra{i}{j}\rb(U\otimes\bar{U})^\dagger=\sum_{i,j=1}^dT\big(\ketbra{i}{j}\big)\otimes \ketbra{i}{j}.
\end{align*}
Restricting to the subgroup of diagonal unitaries in $MU(d,n)$, for which $U^\dagger=\bar{U}$, we have
\begin{align*}
\sum_{i,j=1}^d \rme^{\rmi(\phi_j-\phi_i)}UT\big(\ketbra{i}{j}\big)\bar{U}\otimes \ketbra{i}{j}=\sum_{i,j=1}^dT\big(\ketbra{i}{j}\big)\otimes \ketbra{i}{j}, 
\end{align*}
where $\rme^{\rmi\phi_i}$ is the $i$-th diagonal entry of $U$. Comparing the tensor factors
it follows that
\begin{align}\label{equ:almosteigenvalue}
\rme^{\rmi(\phi_j-\phi_i)}UT\big(\ketbra{i}{j}\big)\bar{U}=T\big(\ketbra{i}{j}\big). 
\end{align}
We will now show that we have
\begin{align}\label{equ:formbeforepermu}
 \tau_T=\sum_{i,j=1}^d A_{ij}\ketbra{i}{i}\otimes \ketbra{j}{j} + B_{ij}\ketbra{i}{j} \otimes \ketbra{i}{j}.
\end{align}
We have
\begin{align*}
T(\ketbra{i}{j})=\sum\limits_{k,l=1}^da_{k,l}\ketbra{k}{l} 
\end{align*}
for some $a_{k,l}\in\bbC$. From \cref{equ:almosteigenvalue} it follows that
\begin{align}\label{equ:writingoutbasis}
\sum\limits_{k,l=1}^d\rme^{\rmi(\phi_k-\phi_l)}a_{k,l}\ketbra{k}{l}=\rme^{\rmi(\phi_i-\phi_j)}\sum\limits_{k,l=1}^da_{k,l}\ketbra{k}{l} 
\end{align}
for all diagonal unitaries.
Again comparing both sides of \cref{equ:writingoutbasis} we have $a_{k,l}\rme^{\rmi(\phi_i-\phi_j)}=\rme^{\rmi(\phi_k-\phi_l)}a_{k,l}$.
Suppose now $i\not=j$. For $a_{k,l}\not=0$ we have
\begin{align}\label{equ:relationphases}
\rme^{\rmi(\phi_i-\phi_j)}=\rme^{\rmi(\phi_k-\phi_l)} 
\end{align}
for all diagonal entries of diagonal unitaries. If $k,l,i$ and $j$ are all pairwise distinct, we have
$i=k$ and $j\not=l$ or $i\not=k$ and $j=l$, then it is clear that we may always
find a combination of $\phi_k,\phi_l,\phi_i$ and $\phi_j$ such that \cref{equ:relationphases} is not satisfied, a contradiction.
For $i=l$ and $k=j$, it is only possible to find such a combination for $n>2$, as otherwise $\phi_i-\phi_j=-(\phi_i-\phi_j)$
always holds.
This proves that we have 
\begin{align}\label{equ:ijeigenvec}
T\lb\ketbra{i}{j}\rb=B_{ij}\ketbra{i}{j}
\end{align}
for $i\not=j$.
For $i=j$ we have analogously that
\begin{align*}
UT(\ketbra{i}{i})\bar{U}=\sum\limits_{k,l=1}^d\rme^{\rmi(\phi_k-\phi_l)}a_{k,l}\ketbra{k}{l}=\sum\limits_{k,l=1}^da_{k,l}\ketbra{k}{l}.  
\end{align*}
In this case, we have $a_{k,l}=\rme^{\rmi(\phi_k-\phi_l)}a_{k,l}$ for all possible phases of the form $e^{\rmi(\phi_k-\phi_l)}$.
It is then clear that $a_{k,l}=0$ unless $k=l$ by a similar argument as before.
This gives 
\begin{align}\label{equ:diagtodiag}
T\lb\ketbra{i}{i}\rb=\sum_{j=1}^d A_{ij}\ketbra{j}{j}. 
\end{align}
Putting together \cref{equ:diagtodiag} and \cref{equ:ijeigenvec} implies \cref{equ:formbeforepermu}.
Next, we will exploit that $\tau_T$  commutes in addition with permutations of the form $U_\pi\otimes U_\pi$ for all $\pi\in S_d$. 
For $i\neq j$ this implies that $A_{i,j}=A_{\pi(i),\pi(j)}$ and $B_{i,j}=B_{\pi(i),\pi(j)}$ so that there is 
only one independent off-diagonal element for each $A$ and $B$. 
The case $i=j$ leads to a third parameter that is a coefficient in front of $\sum_\alpha|ii\rangle\langle ii|$. 
Translating this back to the level of projections then yields \cref{eq:symmchannel}. 
The fact that the terms of \cref{eq:symmchannel} are projections can be seen by direct inspection.
Note that the term corresponding to $\alpha$ is the difference of two projections, the identity and projection onto
diagonal matrices. As the rank of the identity is $d^2$ and the space of diagonal matrices has dimension $d$, we obtain the claim.
The same reasoning applies to the term corresponding to $\beta$, as it is the difference of the projection onto diagonal matrices
and the projection onto the maximally mixed state. The latter is a projection of rank $1$, which yields a rank of $d-1$ for their difference.
\end{proof}

\section*{References}
\bibliographystyle{iopart-num}
\bibliography{generalbenchmarkingliterature}

\providecommand{\newblock}{}
\begin{thebibliography}{10}
\expandafter\ifx\csname url\endcsname\relax
  \def\url#1{{\tt #1}}\fi
\expandafter\ifx\csname urlprefix\endcsname\relax\def\urlprefix{URL }\fi
\providecommand{\eprint}[2][]{\url{#2}}

\bibitem{Poyatos_1997}
Poyatos J~F, Cirac J~I and Zoller P 1997 {\em Phys. Rev. Lett.\/} {\bf 78}
  390--393

\bibitem{Chuang_1997}
Chuang I~L and Nielsen M~A 1997 {\em J. Mod. Opt.\/} {\bf 44} 2455--2467

\bibitem{Knill_2008}
Knill E, Leibfried D, Reichle R, Britton J, Blakestad R~B, Jost J~D, Langer C,
  Ozeri R, Seidelin S and Wineland D~J 2008 {\em Phys. Rev. A\/} {\bf 77}(1)
  012307

\bibitem{Emerson_2007}
Emerson J, Silva M, Moussa O, Ryan C, Laforest M, Baugh J, Cory D~G and
  Laflamme R 2007 {\em Science\/} {\bf 317} 1893--1896

\bibitem{Levi_2007}
L\'evi B, L\'opez C~C, Emerson J and Cory D~G 2007 {\em Phys. Rev. A\/} {\bf
  75}(2) 022314

\bibitem{Emerson_2005}
Emerson J, Alicki R and Zyczkowski K 2005 {\em J. Opt. B\/} {\bf 7} S347

\bibitem{Wallman_Flammia_2014}
Wallman J~J and Flammia S~T 2014 {\em New J. Phys.\/} {\bf 16} 103032

\bibitem{Helsen_2017}
Helsen J, Wallman J~J, Flammia S~T and Wehner S 2017 {\em ArXiv e-prints\/}
  (\textit{Preprint} \eprint{1701.04299})

\bibitem{Chow_2009}
Chow J~M, Gambetta J~M, Tornberg L, Koch J, Bishop L~S, Houck A~A, Johnson B~R,
  Frunzio L, Girvin S~M and Schoelkopf R~J 2009 {\em Phys. Rev. Lett.\/} {\bf
  102}(9) 090502

\bibitem{Ryan_2009}
Ryan C~A, Laforest M and Laflamme R 2009 {\em New J. Phys.\/} {\bf 11} 013034

\bibitem{Olmschenk_2010}
Olmschenk S, Chicireanu R, Nelson K~D and Porto J~V 2010 {\em New J. Phys.\/}
  {\bf 12} 113007

\bibitem{Brown_2011}
Brown K~R, Wilson A~C, Colombe Y, Ospelkaus C, Meier A~M, Knill E, Leibfried D
  and Wineland D~J 2011 {\em Phys. Rev. A\/} {\bf 84}(3) 030303

\bibitem{Gaebler_2012}
Gaebler J~P, Meier A~M, Tan T~R, Bowler R, Lin Y, Hanneke D, Jost J~D, Home
  J~P, Knill E, Leibfried D and Wineland D~J 2012 {\em Phys. Rev. Lett.\/} {\bf
  108}(26) 260503

\bibitem{Barends_2014}
Barends R, Kelly J, Megrant A, Veitia A, Sank D, Jeffrey E, White T~C, Mutus J,
  Fowler A~G, Campbell B, Chen Y, Chen Z, Chiaro B, Dunsworth A, Neill C,
  O/'Malley P, Roushan P, Vainsencher A, Wenner J, Korotkov A~N, Cleland A~N
  and Martinis J~M 2014 {\em Nature\/} {\bf 508} 500--503

\bibitem{Xia_2015}
Xia T, Lichtman M, Maller K, Carr A~W, Piotrowicz M~J, Isenhower L and Saffman
  M 2015 {\em Phys. Rev. Lett.\/} {\bf 114}(10) 100503

\bibitem{Muhonen_2015}
Muhonen J~T, Laucht A, Simmons S, Dehollain J~P, Kalra R, Hudson F~E, Freer S,
  Itoh K~M, Jamieson D~N, McCallum J~C, Dzurak A~S and Morello A 2015 {\em J.
  Phys. Condens. Matter\/} {\bf 27} 154205

\bibitem{Asaad_2016}
Asaad S, Dickel C, Langford N~K, Poletto S, Bruno A, Rol M~A, Deurloo D and
  DiCarlo L 2016 {\em nph Quantum Inf.\/} {\bf 2} 16029

\bibitem{Hashagen_2018}
Hashagen A~K, Flammia S~T, Gross D and Wallman J~J 2018 {\em ArXiv e-prints\/}
  (\textit{Preprint} \eprint{1801.06121})

\bibitem{Brown_Eastin_2018}
Brown W~G and Eastin B 2018 {\em ArXiv e-prints\/} (\textit{Preprint}
  \eprint{1801.04042})

\bibitem{Cross_2016}
Cross A~W, Magesan E, Bishop L~S, Smolin J~A and Gambetta J~M 2016 {\em npj
  Quantum Inf.\/} {\bf 2} 16012

\bibitem{CarignanDugas_2015}
Carignan-Dugas A, Wallman J~J and Emerson J 2015 {\em Phys. Rev. A\/} {\bf
  92}(6) 060302

\bibitem{VandenNest_2011}
Van~den Nest M 2011 {\em New J. Phys.\/} {\bf 13} 123004

\bibitem{Nielsen_2009}
Nielsen M~A and Chuang I~L 2009 {\em Quantum Computation and Quantum
  Information\/} (Cambridge University Press)

\bibitem{Heinosaari_Ziman_2012}
Heinosaari T and Ziman M 2012 {\em The Mathematical Language of Quantum Theory:
  From Uncertainty to Entanglement\/} (Cambridge University Press)

\bibitem{Simon_1996}
Simon B 1996 {\em Representations of finite and compact groups\/} ({\em
  Graduate studies in mathematics\/} vol~10) (American Mathematical Society)

\bibitem{Mendl_Wolf_2009}
Mendl C~B and Wolf M~M 2009 {\em Commun. Math. Phys.\/} {\bf 289} 1057--1086

\bibitem{Mozrzymas_2017}
Mozrzymas M, Studzi\'{n}ski M and Datta N 2017 {\em J. Math. Phys.\/} {\bf 58}
  052204

\bibitem{markovmixing}
Levin D~A and Peres Y 2017 {\em Markov chains and mixing times\/} vol 107
  (American Mathematical Society)

\bibitem{saloff2004random}
Saloff-Coste L 2004 Random walks on finite groups {\em Probability on discrete
  structures\/} (Springer) pp 263--346

\bibitem{Wehrl_1978}
Wehrl A 1978 {\em Rev. Mod. Phys.\/} {\bf 50}(2) 221--260

\bibitem{Nielsen_2002}
Nielsen M~A 2002 {\em Phys. Lett. A\/} {\bf 303} 249--252

\bibitem{Horn_2009}
Horn R~A and Johnson C~R 2009 {\em Matrix Analysis\/} (Cambridge University
  Press)

\bibitem{ergodicchiribella}
{Burgarth} D, {Chiribella} G, {Giovannetti} V, {Perinotti} P and {Yuasa} K 2013
  {\em New J. Phys.\/} {\bf 15} 073045

\bibitem{Wallman_2017}
Wallman J~J 2018 {\em {Quantum}\/} {\bf 2} 47

\bibitem{Proctor_2017}
Proctor T, Rudinger K, Young K, Sarovar M and Blume-Kohout R 2017 {\em Phys.
  Rev. Lett.\/} {\bf 119} 130502

\bibitem{Gambetta_2012}
Gambetta J~M, C\'orcoles A~D, Merkel S~T, Johnson B~R, Smolin J~A, Chow J~M,
  Ryan C~A, Rigetti C, Poletto S, Ohki T~A, Ketchen M~B and Steffen M 2012 {\em
  Phys. Rev. Lett.\/} {\bf 109}(24) 240504

\bibitem{Magesan_2012}
Magesan E, Gambetta J~M and Emerson J 2012 {\em Phys. Rev. A\/} {\bf 85}(4)
  042311

\bibitem{Magesan_2011}
Magesan E, M G~J and Emerson J 2011 {\em Phys. Rev. Lett.\/} {\bf 106} 180504

\bibitem{Dankert_2009}
Dankert C, Cleve R, Emerson J and Livine E 2009 {\em Phys. Rev. A\/} {\bf
  80}(1) 012304

\bibitem{2018arXiv180602048H}
{Helsen} J, {Xue} X, {Vandersypen} L~M~K and {Wehner} S 2018 {\em ArXiv
  e-prints\/} (\textit{Preprint} \eprint{1806.02048})

\bibitem{PerezGarcia_2006}
P{\'e}rez-Garc{\'{i}}a D, Wolf M~M, Petz D and Ruskai M~B 2006 {\em J. Math.
  Phys.\/} {\bf 47} 083506--083506

\bibitem{harrow2009random}
Harrow A~W and Low R~A 2009 {\em Commun. Math. Phy.\/} {\bf 291} 257--302

\bibitem{kliuchnikov2014new}
Kliuchnikov V 2014 {\em New methods for quantum compiling\/} Ph.D. thesis
  University of Waterloo

\bibitem{Leditzky_2018}
Leditzky F, Kaur E, Datta N and Wilde M~M 2018 {\em Phys. Rev. A\/} {\bf 97}

\bibitem{randall2006rapidly}
Randall D 2006 {\em Computing in Science \& Engineering\/} {\bf 8} 30--41

\bibitem{gottesman1997stabilizer}
Gottesman D 1997 {\em ArXiv e-prints\/} (\textit{Preprint} \eprint{9705052})

\bibitem{koenig2014efficiently}
Koenig R and Smolin J~A 2014 {\em J. Math. Phys.\/} {\bf 55} 122202

\bibitem{hokanson2013numerically}
Hokanson J 2013 {\em Numerically stable and statistically efficient algorithms
  for large scale exponential fitting\/} Ph.D. thesis Rice University

\bibitem{Holmstr_m_2002}
Holmstr\"om K and Petersson J 2002 {\em App. Math. Comput.\/} {\bf 126} 31--61

\bibitem{expofitcode}
Hokanson J 2014 Exponential fitting
  \url{https://github.com/jeffrey-hokanson/exponential_fitting_code}

\bibitem{ferriebayesian}
{Hincks} I, {Wallman} J~J, {Ferrie} C, {Granade} C and {Cory} D~G 2018 {\em
  ArXiv e-prints\/} (\textit{Preprint} \eprint{1802.00401})

\bibitem{hoeffding}
Hoeffding W 1963 {\em J. Am. Stat. Assoc.\/} {\bf 58} 13--30

\bibitem{Sason_2015}
Sason I 2015 {\em ArXiv e-prints\/} (\textit{Preprint} \eprint{1503.07118})

\end{thebibliography}

\end{document}